\documentclass[12pt]{article}

\usepackage{amsfonts}

\usepackage{amscd}
\usepackage{amsthm}
\usepackage{indentfirst}

\usepackage{amssymb, amsmath, amsthm, amsgen, amstext, amsbsy, amsopn}
\usepackage{epsfig}

\usepackage{color}
\usepackage{enumerate}
\usepackage{enumitem}

\newtheorem{thm}{Theorem}[section]
\newtheorem{lemma}[thm]{Lemma}
\newtheorem{prop}[thm]{Proposition}
\newtheorem{cor}[thm]{Corollary}

\newtheorem{rem}[thm]{Remark}
\newtheorem{exam}[thm]{Example}

\newcommand{\R}{{\mathbb{R}}}

\newcommand{\Q}{{\mathbb{Q}}}

\newcommand{\Z}{{\mathbb{Z}}}
\newcommand{\N}{{\mathbb{N}}}
\newcommand{\C}{{\mathbb{C}}}

\newcommand{\La}{{\Lambda}}

\newcommand{\cL}{{\mathcal{L}}}

\newcommand{\cS}{{\mathcal{S}}}

\def\id{{1\hskip-2.5pt{\rm l}}}

\newcommand{\supp}{{\rm supp \,}}

\newcommand{\al}{{\alpha}}

\newcommand{\si}{{\sigma}}
\newcommand{\hb}{{\hbar}}

\newcommand{\om}{{\omega}}

\newcommand{\Ga}{{\Gamma}}

\newcommand{\Ci}{{\mathcal{C}}^{\infty}}
\newcommand{\Cl}{\mathcal{C}}
\newcommand{\op}{\operatorname}
\newcommand{\con}{\overline}
\newcommand{\bigo}{\mathcal{O}}
\newcommand{\Hilb}{\mathcal{H}}

\begin{document}

\title{Quantum speed limit vs. classical displacement energy}

\renewcommand{\thefootnote}{\alph{footnote}}

\author{\textsc Laurent Charles  and Leonid Polterovich$^{a}$ }

\footnotetext[1]{Partially supported by the Israel Science Foundation grant 178/13 and the European Research Council Advanced grant 338809.}

\date{\today}

\maketitle

\begin{abstract}
We discuss a  link between symplectic displacement energy, a fundamental notion of symplectic topology, and the  quantum speed limit, a universal constraint on the speed of quantum-mechanical processes. The link is provided by the quantum-classical correspondence formalized within the framework of the Berezin-Toeplitz quantization.
\end{abstract}

\tableofcontents

\section{Introduction and main results}

In the present paper we discuss a seemingly unnoticed link between symplectic displacement energy, a fundamental notion of symplectic topology introduced by Hofer \cite{Hofer} in 1990 and the quantum speed limit, a universal bound for energy required for, roughly speaking, ``orthogonalization" of a quantum state which was discovered by Mandelstam and Tamm \cite{MT} as early as in 1945 and refined by Margolus and Levitin \cite{ML} in 1998. The link is provided by the quantum-classical correspondence formalized within the framework of the Berezin-Toeplitz quantization.

\subsection{Setting the stage}
Given a finite-dimensional complex Hilbert space $\Hilb$, denote by
$\cL(\Hilb)$ the space of all Hermitian operators on $\Hilb$ and by $\cS( \Hilb ) \subset \cL(\Hilb)$ the subset
of all positive trace $1$ operators. In the matrix model of quantum mechanics, elements  $A \in \cL( \Hilb)$
and $\theta \in \cS( \Hilb )$ represent observables and (mixed) states of the quantum system, respectively.
The pure states, i.e. the ones given by rank $1$ projectors, will be identified with the unit vectors
$\xi \in \Hilb $ up to the  phase factor.

In what follows we deal with the measure of
overlap between a pair of quantum states called {\it fidelity}. Recall \cite[p.42]{Hayashi} that  the fidelity $\Phi(\theta,\sigma)$ for  $\theta,\sigma \in \cS(\Hilb)$ is defined as $\|\sqrt{\theta}\sqrt{\sigma}\|_{tr}$, where $\|\cdot\|_{tr}$ stands for the trace
norm. When $\theta=\xi$ and $\sigma=\eta$ are pure states, $\Phi= |\langle \xi, \eta \rangle|$.
Thus, orthogonal quantum states are maximally non-overlapping.

Fix a Planck constant $\hbar >0$. The Schr\"{o}dinger equation with the (time dependent) Hamiltonian $F_t \in \cL( \Hilb) $ is given by
\begin{equation}\label{eq-Schr}
\dot{U}(t) = -\frac{i}{\hbar}F_tU(t)\;,\end{equation}
where $U(t): \Hilb \to \Hilb$ is the unitary evolution with $U(0)=\id$. The Schr\"{o}dinger flow
acts on mixed quantum states by $\theta \to U(t) \theta U(t)^*$ and on pure quantum states by
$\xi \to U(t)\xi$. The total energy of the quantum evolution is given by $\ell_q(F)$,
where by definition
$$\ell_q(F):= \int_0^1 \|F_t\|_{op}dt\;.$$
Here $q$ stands for quantum.

We say that the Schr\"{o}dinger flow $U(t)$ generated by a (time-dependent) quantum Hamiltonian $F$ {\it $a$-dislocates} a state $\theta \in \cS(\Hilb)$, $a \in [0,1)$ if $$\Phi(\theta, U\theta U^{-1}) \leq a\;,$$
where $U$ stands for the time one map $U(1)$ of the flow. Dislocation of states appears in several problems of quantum mechanics and quantum computation. Remarkably,
there exists a universal lower bound for the energy required for dislocation which we call {\it the quantum speed limit}. Take any quantum state $\theta \in \cS(\Hilb)$ and fix $a \in [0,1)$. Then for every
Hamiltonian $F$ which $a$-dislocates $\theta$ one has
\begin{equation}\label{eq-qsl-univ}
\ell_q(F) \geq  \arccos(a)\hbar\;.
\end{equation}
This formula readily follows from the time-energy uncertainty going back to the foundational paper by Mandelstam and Tamm \cite{MT}, and which was proved in the full generality, i.e., for mixed states and  time-dependent Hamiltonians, by Uhlmann \cite{Uh}. It is also closely related to the work of Margolus and Levitin \cite{ML}; in fact the wording ``quantum speed limit" is borrowed from it. We refer to
the paper \cite{AH} by Andersson and Heydari for a comprehensive account, and to Section \ref{subsec-mis-qsl} for a proof of \eqref{eq-qsl-univ} and further comments.

In the present paper we explore semiclassical dislocation of semiclassical states.
Note that manipulating (not necessarily semiclassical) quantum states by applying semiclassical Hamiltonians is one of the themes considered within coherent quantum control theory, see e.g. \cite{Lloyd} and references therein. The problem we address is two-fold. First, what is a speed limit of such a dislocation?
One of our main findings is that {\it on the phase space scales exceeding
the wavelength scale, symplectic topology yields an improvement of inequality
\eqref{eq-qsl-univ}: the speed limit becomes more restrictive than the universal lower bound $\sim \hbar$.} Second, given a quantum state obtained by the quantization of a classical state, we wish to dislocate it by the Schr\"{o}dinger flow corresponding to a classical Hamiltonian
flow. How to design such a Hamiltonian flow? A brief answer, which will be elaborated below, is
that {\it in certain situations it suffices to consider the flows which displace the support of the classical
state in the phase space.}

For the sake of transparency, we start with the case when $(M^{2n},\omega)$ is a quantizable closed K\"{a}hler manifold equipped with the holomorphic line bundle $L$ whose curvature equals $i\omega$. The Berezin-Toeplitz
quantization of $M$ sends smooth real-valued functions $f$ on $M$ to Hermitian operators $T_\hbar(f)$ on the space $\Hilb_\hbar$ of global holomorphic sections of the tensor power $L^k$. Here $k$ is a sufficiently large possible integer and $\hbar=1/k$ is the Planck constant. The operators $T_\hbar(f)$ are given by
$T_\hbar(f):= \int_M f(x) R_\hbar(x) P_{x,\hbar}d\mu(x)$, where $P_{x,\hbar}$ are the rank one {\it coherent states projectors}, $R_\hbar$ is the {\it Rawnsley} function on $M$ and $d\mu$ is the symplectic volume.
For a classical state, i.e., a Borel probability measure $\tau$ on $M$,
define the corresponding quantum state
$$ Q_\hbar(\tau) = \int_M P_{x, \hbar} \; d\tau(x)\;.$$
The states of this form are sometimes called in the physics literature $P$-representable or classical quantum states \cite{GiBrBr}. Here, the adjective ``classical'' refers to the fact that these states are convex combinations of the projectors $P_{x,\hbar}$, and the coherent states are considered as the pure classical states. Nevertheless, these states  are not necessarily classical in the information theoretic sense \cite{FePa}.

The classical counterpart of the quantum dislocation is the notion of displacement.
Consider a Hamiltonian flow $\phi_t$ on $M$ generated by a compactly supported Hamiltonian $f_t$. Its time one map $\phi=\phi_1$ is called {\it a Hamiltonian diffeomorphism}.  {\it Hofer's length} of the path of diffeomorphisms $\{\phi_t\}$, $t \in [0,1]$,  is given
by
$$
\ell_{cl}(f) := \int_0^1\|f_t\|dt \;.
$$
A subset $X \subset M$ is called {\it displaceable} if there exists a Hamiltonian diffeomorphism
$\phi$ displacing $X$, i.e.,  $\phi(X) \cap X = \emptyset$. The {\it displacement energy} $e_M(X)$ is defined as $\inf \ell_{cl}(f)$, where the infimum is taken over all Hamiltonians whose time one map
displaces $X$. A fundamental result of symplectic topology  states that $e_M(X) >0$ for every displaceable open subset of $M$, see \cite{Hofer,LMc}. Let us mention that $e_M(\cdot)$ is monotone with respect
to inclusion of subsets, i.e.,
\begin{equation}\label{eq-dem}
e_M(X) \leq e_M(Y) \;\;\forall X \subset Y\;.
\end{equation}
Furthermore, for subsets $X$ of the symplectic linear space $\R^{2n}$ equipped with a translation invariant symplectic form, the displacement energy scales under dilations as the 2-dimensional area:
\begin{equation} \label{eq-des}
e_{\R^{2n}}(s\cdot X) = s^2e_{\R^{2n}}(X) \;\; \forall s >0\;,
\end{equation}
where
\begin{equation}\label{eq-dil}
s \cdot X := \{sx\;|\; x \in X\}\;.
\end{equation}
We refer to \cite{P-Rosen-book} for preliminaries on Hofer's geometry.

\subsection{Dislocation yields displacement} Now we are ready to formulate our first main result
which will be proved in Section \ref{sec:proof-theorem-1.2ii} below.

\medskip
\noindent
\begin{thm}\label{thm-intro-main-1} Let $\tau $ be a classical state on $M$ with
$d \tau = u d\mu$, where the density $u$ is of class $\Cl^3$.
Let $f_{t,\hbar}$ be a family of Hamiltonians which are uniformly (in $\hbar$) bounded
together with $4$ derivatives. Suppose that the  Schr\"{o}dinger flow generated by the quantized Hamiltonian $F_{t,\hbar} = T_\hbar(f_{t,\hbar})$  $o(\hbar^n)$-dislocates the state $Q_\hbar(\tau)$. Then for
every $\lambda>0$, when $\hbar$ is sufficiently small, the Hamiltonian flow generated by $f_{t,\hbar}$ displaces the subset $\{u > \lambda\}$ and
\begin{equation}
\label{eq-ellq-main}
\ell_q(F_{\hbar}) \geq e_M(\{u > \lambda\})+ \bigo(\hbar)\;.
\end{equation}
\end{thm}

\medskip
\noindent
In particular, the speed limit for the semiclassical dislocation is $\sim 1$ which
is more restrictive than the universal quantum speed limit $\sim \hbar$. Let us mention also
that by \eqref{eq-dem} $e_M(\{u > \lambda\})$ is a non-increasing function of $\lambda$.

In Section \ref{sec-ddsc} below we zoom into small scales exceeding the quantum length scale $\sqrt{\hbar}$.
In particular, we extend Theorem \ref{thm-intro-main-1} to dislocation of semiclassical states which, roughly speaking, occupy a ball of radius  $\hbar^\varepsilon$, $ \varepsilon \in [0,1/2)$ in the phase space. We  show that the speed limit on such a scale is of the order $\sim \hbar^{2\varepsilon}$ which, again, is more restrictive than the universal quantum speed limit.

\begin{rem} {\rm
Let us mention that working on small scales is technically challenging: we use sharp remainder bounds for the Berezin-Toeplitz quantization elaborated in \cite{our_paper} and the arXiv version of \cite{oim_symp}.  These bounds, which
involve (higher) derivatives of classical observables and the Planck constant $\hbar$,  are optimized in such a way that they behave in a friendly manner under rescaling of
classical observables. We refer to Section \ref{sec:semiclassical-limit} and formula (\ref{eq:Expansion_complement}) of section \ref{sec-ddsc}. }
\end{rem}

As an immediate consequence of Theorem \ref{thm-intro-main-1}, we get the following symplectic obstruction to $o(\hbar^n)$-dislocation:

\medskip
\noindent
\begin{cor}\label{cor-nodis} If the set $\{u > \lambda\}$ is non-displaceable for some $\lambda >0$, the quantum state $Q_\hbar(ud\mu)$ cannot be $o(\hbar^n)$-dislocated by a semiclassical Hamiltonian
of the form $T_\hbar(f_t)$.
\end{cor}

\begin{exam}\label{exam-CP}{\rm Let $M= \C P^n$ be the complex projective plane equipped with the Fubini-Study
symplectic form whose integral over a complex line equals $2\pi$.  Assume that $\{u >0\}$ contains a closed symplectically embedded ball $B$ whose volume is $\geq \text{Vol}(\C P^n)/2^n$. Then $B$ is non-displaceable
by Gromov's packing theorem \cite{Gromov}, and hence the quantum state $Q_\hbar(ud\mu)$ cannot be $o(\hbar^n)$-dislocated by a semiclassical Hamiltonian of the form $T_\hbar(f_t)$. As a counterpoint,
we shall show below that it can be $\epsilon$-dislocated by such a Hamiltonian with an arbitrary $\epsilon >0$ provided the support of $u$ has volume $< \text{Vol}(\C P^n)/2$.}
\end{exam}

\subsection{Displacement yields dislocation}\label{subsec-dyd}
For any  family $\theta = (\theta_{\hbar}  \in S( \Hilb_{\hbar}))_{\hbar}$ of quantum states,  consider   Borel probability measures $(\nu_{\hbar})$ on $M$ defined by $\int_M f d \nu_{\hbar} = \op{tr} ( T_{\hbar} (f) \theta_{\hbar})$. These measures, which will be called {\it Husimi measures}, govern the phase space distribution of the quantum states. Their limiting
behaviour as $\hbar \to 0$ have been attracting a lot of attention in semiclassical analysis.
An important characteristic of this behaviour is {\it the microsupport} of $\theta$. It can be defined similarly as the semiclassical wavefront set by using Toeplitz operators instead of $\hbar$-pseudodifferential operators. Equivalently, it is the subset $\op{MS} ( \theta)$ of $M$ such that $$ x \notin \op{MS} ( \theta) \Leftrightarrow \text{ $x$ has a neighborhood $U$ such that $\nu_{\hbar} ( U) = \bigo ( \hbar^{\infty})$ }\;.$$ Here and below we say that a sequence $a_\hbar \in \R$ is $\bigo(\hbar^\infty)$ if for every $N \in \N$ \footnote{We denote by by $\N$ the set of non negative integers and let $\N^{*} = \N \setminus \{ 0 \}$.} there exists
a constant $C_N$ such that $|a_\hbar| \leq C_N\hbar^N$ for all sufficiently small $\hbar$. Thus, the sequence
of measures $\nu_\hbar$ ``rapidly dissipates" outside the microsupport.

\begin{exam}\label{exam-MS}{\rm  For a classical state (i.e., a measure) $\tau$ on $M$ , the microsupport
of its quantum counterpart coincides with its support: $\op{MS}(Q_\hbar(\tau)) = \text{supp}\; \tau$,
see Section \ref{subsec-mscm} below.}
\end{exam}

\noindent The next result establishes a link between the notion of microsupport
and fidelity.

\begin{prop}\label{prop-MS-fidelity} Let $\theta= (\theta_\hbar)$ and $\sigma= (\sigma_\hbar)$
be two families of quantum states with disjoint microsupports: $\op{MS}(\theta)\cap \op{MS}(\sigma)=\emptyset$. Then $\Phi(\theta_\hbar,\sigma_\hbar)= \bigo(\hbar^\infty)$.
\end{prop}

\noindent
Another basic fact is that the microsupport propagates along the Hamiltonian flow:
Let $f_t$ be a Hamiltonian on $M$, and let $F_{t,\hbar} = T_\hbar(f_t)$ be the corresponding quantum
Hamiltonian. Write $\phi$ for the time one map of the Hamiltonian flow of $f_t$ and
$U_\hbar$ for the time one map of the Schr\"{o}dinger evolution of $F_{t,\hbar}$. For any sequence
of quantum states $\theta= (\theta_\hbar)$ put $\theta':= (U_\hbar \theta_\hbar U_\hbar^{-1})$.
Then
\begin{equation} \label{eq-propag-MS} \op{MS}(\theta')= \phi(\op{MS}(\theta))\;.
\end{equation}
The proofs of Proposition \ref{prop-MS-fidelity} and equation \eqref{eq-propag-MS} will be given in Section \ref{subsec-dp}. Combining these two results, we get
that displacement of the microsupport yields dislocation.

\medskip
\noindent
\begin{thm}\label{thm-former-main-i} Let $\theta = (\theta_{\hbar}  \in S( \Hilb_{\hbar}))_{\hbar}$ and $(f_t)$ be a time-dependent Hamiltonian of $M$. Assume that the time-one-map of the Hamiltonian flow generated by $f_t$ displaces the microsupport of $\theta$. Then the time-one-map of  the Schr\"{o}dinger flow generated by  $ T_\hbar(f_{t})$ $\bigo(\hbar^\infty)$-dislocates the state $\theta$.
\end{thm}

\medskip
\noindent In view of Example \ref{exam-MS} we get the following corollary.

\medskip
\noindent
\begin{cor}\label{thm-intro-main-1-i} Let $\tau $ be a classical state on $M$.
Assume that the time-one-map of the Hamiltonian flow generated by a Hamiltonian $f_t$ displaces the support of $\tau$. Then  the Schr\"{o}dinger flow generated by the quantized Hamiltonian $ T_\hbar(f_{t})$ $\bigo(\hbar^\infty)$-dislocates the state $Q_\hbar(\tau)$.
\end{cor}

\medskip
\noindent
In particular, the class of Hamiltonians appearing in Theorem \ref{thm-intro-main-1} above
is non-empty provided the support of $\tau$ is displaceable.

\subsection{Symplectic rigidity vs. flexibility on the quantum side}
Exploration of the borderline  between symplectic rigidity and symplectic flexibility is one
of central themes in symplectic topology. For instance, existence
of non-trivial non-displaceable subsets, i.e., subsets which cannot be displaced by a Hamiltonian flow but can be displaced by a volume-preserving flow,  is a manifestation of symplectic rigidity which has been intensively studied since the 1980-ies, and which plays an important role in our discussion.
Think for instance, about the symplectically embedded ball $B \subset \C P^n$ from
Example \ref{exam-CP} above: it has the above-mentioned features provided its volume satisfies
$$  \text{Vol}(\C P^n)/2^n < \text{Vol}(B) < \text{Vol}(\C P^n)/2\;.$$

This phenomenon becomes especially surprising if one confronts it with the following statement
which reflects symplectic flexibility: any open subset $U$ of a symplectic manifold $M$ occupying
less than half of its volume can be displaced by a Hamiltonian flow {\it up to an arbitrary small
measure}: for every $\epsilon >0$ there exists a flow $\phi_t$ with $\text{Vol}(\phi_1U \cap U ) < \epsilon$. Furthermore, while the energy $\ell_{cl}$ required to displace an open subset of a symplectic manifold is bounded away from zero (this is another facet of symplectic rigidity), the displacement up to a small measure can be performed with an arbitrary small energy. The matrix inequality $\|\cdot\|_{op} \leq \|\cdot\|_{tr}$, which has no classical analogue, enables one to use this flexibility phenomenon for producing meaningful dislocation regimes
of semiclassical quantum states with small values of quantum energy $\ell_q$.

\medskip
\noindent
\begin{thm}\label{thm-main-31}
Let $d\tau =ud\mu$ be a classical state on $M$. Suppose that the volume of the support of $u$ is $< \text{Vol}(M)/2$. Then for every
$\epsilon >0, \delta >0$ the state $\theta_\hbar:= Q_\hbar(\tau)$ can be $\epsilon$-dislocated with
$\ell_q \leq \delta$ by the Schr\"{o}dinger flow generated by $T_\hbar(f_t)$ for an appropriately
chosen $\hbar$-independent Hamiltonian $f_t$, for all $\hbar$ small enough.
\end{thm}

\medskip
\noindent
The proof is given in Section \ref{subsec-thm-main-31-pf} below.
This result, when compared with Theorem \ref{thm-intro-main-1} and Corollary \ref{cor-nodis},  highlights
a drastic contrast between $o(\hbar^n)$-dislocation and $\epsilon$-dislocation with $\epsilon >0$ independent of $\hbar$, which reflects the quantum counterpart of the interplay between symplectic rigidity and flexibility. For instance, consider the quantum state $\theta_\hbar = Q_\hbar(ud\mu)$ from Example \ref{exam-CP} above such that the support of the density $u$ occupies less than half of the volume of $\C P^n$ while the set $\{u > 0\}$ contains a non-displaceable ball. It follows that $\theta_\hbar$ cannot be $o(\hbar^n)$-dislocated by a Hamiltonian of the form $T_\hbar(f_t)$, but it admits an $\epsilon$-dislocation
by such a Hamiltonian with arbitrarily small $\epsilon$. In other words, {\it the contrast between symplectic rigidity and flexibility on the quantum side is reflected in the  asymptotic behavior of the magnitude of dislocation as $\hbar \to 0$.}

\subsection{Alternative mechanisms of dislocation}
Let $f_{t,\hbar}$ be a Hamiltonian on $M$. For a sequence of quantum states $\theta= (\theta_\hbar)$
denote by $\theta'$ its image under the time one map of the corresponding Schr\"{o}dinger evolution
generated by $T_\hbar(f_{t,\hbar})$. Recall a basic mechanism of dislocation provided by Proposition \ref{prop-MS-fidelity} above: the quantum Hamiltonian $\bigo{(\hbar^\infty)}$-dislocates $\theta$ if
the microsupports of $\theta$ and $\theta'$ are disjoint: $\op{MS}(\theta) \cap \op{MS}(\theta')=\emptyset$. We design two other mechanisms of $\bigo{(\hbar^\infty)}$-dislocation which
demonstrate that, in general, the opposite is not true. Let us outline their main features leaving
a detailed description for Section \ref{sec:other-exampl-disl} below.

Our first example illustrates, in addition,  that the implication ``dislocation yields displacement" provided by Theorem \ref{thm-intro-main-1} above is specific for ``classical" quantum states of the form $Q_\hbar(\tau)$ and in general does not hold for other ``natural" families of quantum states.
Here $f_{t,\hbar}$ is of the form $\hbar^\delta g_t$ with $\delta \in (0,1/2)$.
The sequences of the Husimi measures corresponding to $\theta$ and $\theta'$ weakly converge to the same measure $\tau$ on $M$. The microsupports of $\theta$ and $\theta'$ both coincide with the support of $\tau$.  Furthermore, $\supp {\tau}$ has a non-empty interior. Given any open subset of $\supp {\tau}$,
its displacement energy is positive. Since the total classical energy $\ell_{cl} = o(1)$,
the flow of $f_{t,\hbar}$ does not displace it for sufficiently small $\hbar$. This contrasts with Theorem \ref{thm-intro-main-1} which prohibits such a behavior for ``classical" quantum states.

The construction goes, roughly speaking, as follows. We start
with an open subset $U \subset M$ and fix a grid $\Gamma_\hbar$  in $U$ with the mesh $s = \hbar^\delta$.
The states $\theta_\hbar$ are chosen  as a normalized superposition $\sum_{x \in \Gamma_\hbar} \varphi(x)\xi_{x,\hbar}$  of the coherent states located at the points of the grid. Here $\varphi$ is an $\hbar$-independent function whose support lies in $U$. The Hamiltonian $f_{t,\hbar}$ displaces
the grid $\Gamma_\hbar$ by half-mesh in some direction. This guarantees convergence of the Husimi measures with $\tau= \varphi^2 d\mu$. Here we may or may not consider that this dislocation comes with a displacement, according to the scale we look at. In fact, the displacement takes place on the scale of the mesh
$\sim \hbar^\delta$, but it is invisible at the constant scale.
We refer to Section \ref{subsec-scc} for the precise formulations.

In the second example, $f_{t,\hbar}$ is of the form $\hbar g_t + \bigo(\hbar^2)$, so the dislocation
regime reaches the quantum speed limit $\sim \hbar$. The microsupports of $\theta$ and $\theta'$ coincide and form a closed Lagrangian submanifold $N \subset M$. In particular, their  displacement energy is positive \cite{Chekanov}.  The states $\theta_\hbar$ are Lagrangian states associated to $N$ in the sense of \cite{oim2}. Roughly speaking, they correspond to holomorphic sections of the quantum line bundle sharply concentrated on $N$. The Hamiltonian $f_{t,\hbar}$ is constructed to produce a destructive interference of the state $\theta$ and its image under the Schr\"{o}dinger evolution, see Section \ref{subsec-lasd} for details.

\medskip
\noindent
{\bf Organization of the paper:}
In Section \ref{sec-qprelim} we collect necessary preliminaries
on the Berezin quantization. Section \ref{sec-wdyd} contains the proof of Theorem \ref{thm-intro-main-1}
(``dislocation yields displacement"). In Section \ref{sec-MS} we study microsupports and deduce its
basic properties which were used above for showing that ``displacement yields dislocation".
Section \ref{subsec-thm-main-31-pf} contains the proof of Theorem \ref{thm-main-31} (a quantum counterpart of symplectic flexibility). In Section \ref{sec:other-exampl-disl} we present alternative mechanisms
of dislocation and, in particular, discuss Lagrangian states. In Section \ref{sec-ddsc} we zoom into small scales and extend our results on dislocation vs. displacement in this setting. We conclude with some auxiliary proofs and comments.

\section{Quantization}\label{sec-qprelim}
In this section we collect preliminaries on various aspects of the Berezin quantization and its semiclassical limit
which are used in the present paper.

Let $(M^{2n},\omega)$ be a symplectic closed manifold. The Hamiltonian vector field $X_f$
of a function $f$ is defined by $i_{X_f}\omega + df = 0$, and the Poisson bracket is given by $\{f,g\}= -\omega(X_f,X_g)$. We denote by $\mu$ the measure of $M$ defined by the volume form $\omega^n/n!$.

Fix an auxiliary Riemannian metric $\rho$ on $M$. For a function
$f \in \Cl^\infty(M)$ its $\Cl^k$-norm with respect to $\rho$ is denoted by
$|f|_{k}$. For a pair of smooth functions $f,g$ put
$|f,g|_{N} = \sum_{j=0}^N |f|_{j} \cdot |g|_{N-j}$. Denote
$$
|f,g|_{1,3} := |f|_{1} \cdot |g|_{3} +  |f|_{2} \cdot |g|_{2}+ |f|_{3} \cdot |g|_{1}\;.
$$
We write $\|f\|=|f|_0 =\max |f|$ for the uniform norm, and $\|f\|_{L_1}$ for the $L_1$-norm of $f$ with respect to $\mu$.

\subsection{Basic notions}

Let $L \rightarrow M$ be a Hermitian line bundle. Define the scalar product of two sections $s, t$ of $L$ by $\langle s, t \rangle = \int_M (s,t) (x) d \mu (x)$ where $(s,t)$ is the pointwise scalar product. So the norm is
$$ \| s \| = \Bigl( \int_M |s (x)|^2 \; d\mu (x) \Bigr)^{1/2}$$
where $|s(x)|$ is the poinwise norm of $s$ at $x$.
Let $\mathcal{H}$ be a finite dimensional subspace of $\Ci ( M , L)$, the {\em quantum space}. To each function $f \in \Cl ^0 (M)$ corresponds the {\em Toeplitz operator} $T(f): \mathcal{H} \rightarrow \mathcal{H}$ defined by  $$\langle T(f) s, t \rangle = \langle f s, t \rangle, \qquad  \forall s, t \in \mathcal{H}.$$ By duality, to any state $\theta \in \cS ( \Hilb)$ is associated a Borel probability measure $\nu$ of $M$ such that $\int_M f d \nu = \op{trace} ( T(f) \theta)$, called the {\em Husimi measure} of $\theta$.

We shall denote by $\con{L}$ the line bundle $L$ equipped with the complex structure
$$\lambda \cdot_{\con{L}} \xi :=  \con{\lambda} \cdot_L \xi \;\;\forall \lambda \in \C, \xi \in L\;,$$
where $\cdot_{\con{L}}$ and $\cdot_L$ stands for the complex multiplication in $\con{L}$ and
$L$, respectively. For an element $e \in L$ (respectively, $e \in \con{L}$) , we write $\con{e}$ for the corresponding element of $\con{L}$ (respectively, of $L$). This is done for convenience only, as $L$ and $\con{L}$ coincide as sets and hence, formally speaking, $e=\con{e}$.

The Hermitian structure on $L$ enables one to identify $\con{L}$ with the dual bundle to $L$.
We have the following natural ``contraction" maps which, by a slight abuse of notation, are again denoted by
$\langle \cdot, \cdot \rangle$ and $(\cdot, \cdot )$:

$$\mathcal{H} \otimes (\mathcal{H} \otimes \con{L}_x) \to L_x,  \;\; \langle s,t \otimes e\rangle :=
\langle s,t \rangle \con{e}\;,$$
and
$$L_x \otimes (\mathcal{H} \otimes L_x) \to \mathcal{H}, \;\; (d,t \otimes e) := (d,e)t\;.$$

For any $x\in M$, the {\em coherent state at $x$} is the vector  $e_x \in \mathcal{H} \otimes \con{L}_x$ such that  for any $s \in \mathcal{H}$, $ s(x) = \langle s , e_x \rangle$. The Toeplitz operator $T(f)$ may be equivalently defined by
\begin{gather} \label{eq:defToep_eq}
T (f) = \int_M f(x)  S_x d\mu (x)
\end{gather}
where $S_x $ is the operator $\mathcal{H} \rightarrow \mathcal{H}$ defined by  $S_x (s) = (s(x) , \con{e}_x)$.

Assume that for any $x$, $e_x  \neq 0 $. \footnote{For the standard K\"{a}hler
quantization discussed in the next section this follows from the Kodaira embedding theorem provided
the Planck constant is sufficiently small.} Thus it has the form $\xi_x \otimes \con{v_x}$ where
$v_x \in L_x$ and $\xi_x$ is a unit vector in $\mathcal{H}$ (here we use that $\con{L}_x$ is one-dimensional). We call $\xi_x$ {\it the normalized coherent state}. It is defined up to a phase factor. We write $P_x$ for the rank one {\it coherent state projector} onto the line  spanned by $\xi_x$.
Note that $P_x = \| e_x \|^{-2} S_x $, where $\|\cdot \|$ stands for the natural norm on the
tensor product $\mathcal{H} \otimes \con{L}_x$ coming from the Hermitian structures on the factors.
In the sequel we will denote by $R$  {\it the Rawnsley function} $R(x) = \| e_x \|^2$ on $M$.
With this notation we have
\begin{equation}
\label{eq-povm-rawn}
T(f) = \int_M f(x)R(x)P_xd\mu(x)\;.
\end{equation}

For any Borel probability measure $\tau$ of $M$, define the ``classical" quantum state
$$ Q(\tau) := \int_M P_x d \tau (x) \in \cS( \Hilb )\;.$$
For any $f \in \Cl^0 (M)$, let $\mathcal{B}(f)$ be the {\em Berezin} transform of $f$ defined by\begin{gather} \label{eq:def_Berezin_transform}
\mathcal{B}(f) (x) =  \langle f e_x , e_x \rangle / \| e_x \|^2.
\end{gather}
Since $\mathcal{B}(f) (x) = \op{trace} (T(f) P_x)$, we have
\begin{gather} \label{eq:duality}
 \op{trace} (T(f) Q( \tau) ) = \int_M \mathcal{B}(f) d\tau\;.
\end{gather}
In other words, the Berezin transform and the map sending $\tau$ to the Husimi measure of $Q( \tau)$ are adjoint to each other.

 Let us mention that instead of considering  first the line bundle $L$ and a subspace of $\Ci ( M , L)$, we could have started with the following data:
\begin{itemize}
\item[-] a finite dimensional Hilbert space $\mathcal{H}$;
\item[-] a smooth family of rank one projector $(P_x: \Hilb \rightarrow \Hilb,\; x \in M)$;
\item[-]  a positive function $R \in \Ci (M)$ such that
\begin{equation}\label{eq-star}
\int_M R(x) P_x d \mu (x) = \op{id}_\Hilb\;.
\end{equation}
\end{itemize}
We can then define the line bundle $L$ so that $L_x$ is the image of $P_x$. We have a natural map $\Phi$ from $\Hilb$ to $\Ci ( M, L)$ sending $s$ to the section whose value at $x$ is $s(x) = (R(x))^{1/2} (P_x s)$. By equation \eqref{eq-star}, $\Phi$ is an isometry. In particular $\Phi$ is injective and we can identify $\Hilb$ with the image of $\Phi$. If $e_x$ is the coherent state in the image of $\Phi$ defined as above, we have that  $R(x) = \| e_x \|^2$ and $P_x$ is the projector onto the line determined by $e_x$.

\subsection{Berezin quantization} \label{sec:semiclassical-limit}

Let $\Lambda \subset \R_{>0}$ be a subset having $0$ as a limit point. Suppose that for any $\hbar \in \Lambda$ we are given a Hermitian line bundle $L_{\hbar}  \rightarrow M$ and a finite dimensional subspace $\Hilb_{\hbar}$  of $\Ci ( M , L_\hbar)$.  Let $T_\hbar: \Ci(M) \to \cL(\Hilb)$ be the map sending a function $f$ to
the corresponding Toeplitz operator. We call the family $T_\hbar$, $\hbar \in \Lambda$  a {\em Berezin} quantization of $M$
provided it satisfies certain list of properties reflecting the mathematical formalism behind the quantum-classical correspondence. In this section we present the package fitting our purposes.

A typical example is a closed K\"ahler manifold with a positive line bundle $L \rightarrow M$. In this case, we set $\Lambda = \{ k^{-1}, k \in \N^* \}$ and define $\Hilb_{\hbar}$ as the space of holomorphic sections of $ L^k$ with $k = \hbar ^{-1}$. Here and below we abbreviate $L^k= L^{\otimes k}$.
More generally, for any closed symplectic manifold $M$ such that $[\om ]/(2\pi) \in H^2 ( M , \Z)$, consider a family $(L_{\hbar}, \Hilb_{\hbar}, \hbar \in \Lambda)$ defined as in \cite{oim_symp}.

The main property of these quantizations is that $\langle e_{x,\hbar} , e_{y, \hbar} \rangle $ has a specific behavior in the limit $\hbar \rightarrow 0$. First, for $x=y$, the Rawnsley function
satifies
\begin{gather} \label{eq:estime_diag} \tag{CS1}
R_\hbar(x):=  \| e_{x, \hbar} \|^2 = (2 \pi \hbar)^{-n} ( 1 + \bigo ( \hbar))
\end{gather}
where the $\bigo $ is uniform in $x$. This also holds in $\Ci  $ topology, meaning that the successive derivative of $\| e_{x, \hbar} \|^2$ are uniformly $\bigo ( \hbar ^{ -n +1})$. As an immediate consequence of this we get from \eqref{eq-povm-rawn} (substituting $f=1$) that
$$\dim \Hilb_\hbar= (2 \pi \hbar)^{-n} ( 1 + \bigo ( \hbar))\;.$$

Second, for any disjoint closed sets $X$, $Y$ of $M$, we have
\begin{gather} \label{eq:estime_outside_diag}  \tag{CS2}
 \langle e_{x,\hbar} , e_{y, \hbar} \rangle = \bigo ( \hbar ^{\infty})
\end{gather}
uniformly in $x \in X$ and $y \in Y$.
Actually  $\langle e_{x,\hbar} , e_{y, \hbar} \rangle $ has a precise  asymptotic expansion uniform in $x$ and $y$, which determines it modulo $\bigo ( \hbar ^{\infty})$. In the case of K\"ahler manifold, this is a deep result coming from \cite{BoSj}. In the case of symplectic manifold, the spaces $\Hilb_{\hbar}$ are defined in such a way that this asymptotic expansion holds. All the semiclassical results we will give are consequences of this expansion.

In particular, the Toeplitz operators form a semiclassical algebra.
More precisely, there exists bidifferential (i.e., differential with respect to each variable)  operators $(B_{\ell}$, $\ell \in \N$) from  $\Ci (M) \times \Ci (M)$ to $\Ci (M)$ such that
$$B_0 (f,g) = fg, \qquad B_1(f,g) - B_1 (g,f) = i \{ f, g \}$$ and for any $f,g \in \Ci (M)$ and $N \in \N$, we have the expansion
\begin{gather} \label{eq:expansion_product} \tag{E}
 T_\hbar(f) T_\hbar(g) = \sum_{\ell = 0 }^{N} \hbar^{\ell} T_\hbar( B_{\ell} (f,g) ) + r_{N,\hbar} ( f,g)
\end{gather}
where $r_{N, \hbar} (f,g) = \bigo ( \hbar^{ N+1})$  in the operator norm. For K\"ahler manifold, this has been deduced in \cite{BoMeSc} from \cite{BoGu}. For symplectic manifold, this is proved in \cite{oim_symp}. Other references with similar results are \cite{Gu},\cite{BU} and \cite{MaMa}.

We clearly have $\| T_{\hbar} (f) \|_{op} \leqslant \| f \|$. This upper bound is actually sharp in the semiclassical limit, more precisely we have the following G\"arding estimate
\begin{gather} \label{eq:P1} \tag{P1}
  \|f\|- \alpha |f|_{2}\hbar \leq \|T_\hbar(f)\|_{op} \leq \|f\|
\end{gather}
where $\al$ is a positive constant.
It has been proved by the authors in \cite{our_paper} that the remainders in (\ref{eq:expansion_product}) satisfy $r_{1, \hbar} ( f,g) = \bigo ( \hbar) |f,g|_{2}$ and $r_{2, \hbar} (f,g)  - r_{2, \hbar}  ( g,f) = \bigo ( \hbar^2) |f,g|_{1,3}$. Consequently, there exist constants $\beta$ and $\gamma$ such that
\begin{gather} \label{eq:P2}
\tag{P2} \qquad \| -\frac{i}{\hbar} \cdot [T_{\hbar}(f),T_{\hbar}(g)] - T_\hbar (\{f,g\})\|_{op} \leq \beta \cdot |f,g|_{1,3} \hbar\;; \\
\label{eq:P3} \tag{P3}  \|T_\hbar(fg) - T_\hbar(f)T_\hbar(g)\|_{op} \leq \gamma |f,g|_{2}\hbar
\end{gather}
The number of derivatives in the bounds of (\ref{eq:P1}), (\ref{eq:P2}) and (\ref{eq:P3}) will be important in Section \ref{sec-ddsc} when we will consider small scales.
As a last property, the Berezin transform satisfies
\begin{gather} \label{eq:Berezin_transform} \tag{B}
 \mathcal{B}_{\hbar} ( f) = f + \bigo (\hbar) | f |_2
\end{gather}
Except in Sections \ref{sec:other-exampl-disl}  and \ref{sec:proof-theorem-disloc_21},  (\ref{eq:estime_diag}), (\ref{eq:estime_outside_diag}), (\ref{eq:expansion_product}), (\ref{eq:P1}), (\ref{eq:P2}), (\ref{eq:P3}) and (\ref{eq:Berezin_transform})  are the only properties we will use.

As a consequence of (\ref{eq:estime_diag}) and (\ref{eq:Berezin_transform}), the trace norm of a Toeplitz operator is estimated by
\begin{gather} \label{eq:tracenorm}
         \| f \|_{L_1} + \bigo( \hbar ) |f|_2 \leqslant  (2 \pi \hb )^n \| T_{\hb}(f) \|_{tr} \leqslant \| f \|_{L_1} ( 1 +   \bigo ( \hbar))
\end{gather}
for any $f \in \Cl^ 2 (M)$.
The proof is given in Section \ref{sec:estimate-trace-norm}.
The lower bound could be written differently. Actually we have that $ \| \mathcal{B}_\hbar(f) R_\hbar \|_{L_1} \leqslant  \| T_{\hbar} (f) \| _{\op{tr}}$ where  $R_\hbar$ is the Rawnsley function.

\subsection{The state $Q_{\hbar} ( \tau)$ in the semiclassical limit} \label{sec:state-semiclassical-limit}

Consider a classical state $\tau$ and the corresponding family of quantum states $Q_{\hbar} ( \tau)$. Let $\nu_{\hbar}$ the semiclassical measure of $Q_{\hbar}(\tau)$.  By the duality (\ref{eq:duality}) and the property (\ref{eq:Berezin_transform}) of the Berezin transform, $\nu_{\hbar}$ converges weakly to $\tau$. More precisely for any $f$ of class $\Cl^2$,
$$  \int_{M} f d \nu_{\hbar} = \op{trace} (T_{\hbar} (f) Q_{\hbar} ( \tau) )  = \int_M f d \tau + \bigo ( \hbar) |f|_2 .$$
By this result, we may interpret $Q_{\hbar} ( \tau)$ as a quantization of $\tau$.

Consider now classical states $\tau_1$, $\tau_2$ of $M$ admitting continuous densities $g_1$, $g_2$ with respect to $\mu$ and such that the square roots of $g_1$, $g_2$ are of class $\Cl^2$. Then the fidelity of $\theta_1 = Q_{\hbar} ( \tau_1)$ and $\theta_2 = Q_{\hbar} ( \tau_2)$ satisfies
\begin{gather} \label{eq:fidelity_estimate}
 \Phi ( \theta_1, \theta_2 ) =   \int_M  (g_1 g_2 )^{1/2} d \mu + \bigo ( \hbar^{1/2})
\end{gather}

\begin{proof}[Proof of Equation (\ref{eq:fidelity_estimate})]
Since the Rawnsley function $R$ is positive, $R^{-1/2}$ is smooth and by (\ref{eq:estime_diag}) $R^{-1/2} = ( 2 \pi \hb)^{n/2} ( 1 + \bigo ( \hbar ))$ in the $\Ci$ topology.  Denoting by $f_i$ the square root of $g_i$, we have by (\ref{eq:P3}) that
$$ \bigl( T_{\hbar} ( f_i R^{-1/2} ) \bigr) ^2 = \theta_i + \bigo ( \hbar ^{n+1})$$
Recall now \cite{Ando} that the square root is an operator monotone function and hence
for any positive Hermitian matrices $A$ and $B$
$$\|\sqrt{A}-\sqrt{B}\|_{op} \leq \|\sqrt{|A-B|}\|_{op}\;.$$
Consequently, $\theta_i ^{1/2}=  T_{\hbar} ( f_i R^{-1/2} ) +  \bigo ( \hbar ^{(n+1)/2})$. We deduce from (\ref{eq:P1}), (\ref{eq:P3}) that
$$\theta^{1/2}_1 \theta^{1/2}_2 = T_{\hbar} ( f_1 f_2  R^{-1}) + \bigo ( \hb^{n+\frac{1}{2}})$$
and using (\ref{eq:tracenorm}), (\ref{eq:estime_diag}) and the fact that the dimension of $\mathcal{H}_\hbar$ is a $\bigo ( \hbar^{-n})$, the result follows.
\end{proof}

The assumption that the square roots of $g_1$ and $g_2$ are of class $\Cl^2$ is rather restrictive, especially when $g_1$ and $g_2$ vanish at some points. We can actually show that for any continuous densities $g_1$, $g_2$
\begin{gather} \label{eq:fidelity_estimate_continuous}
 \Phi ( \theta_1, \theta_2 ) =   \int_M  (g_1 g_2 )^{1/2} d \mu + o (1)
\end{gather}
The proof is the same as the one of (\ref{eq:fidelity_estimate}), but we have to adapt (P1), (P3) and (B) to continuous functions, by relaxing the remainders to $o(1)$ instead of $\bigo( \hbar)$. Since (\ref{eq:fidelity_estimate_continuous}) has no application in this paper, we won't give more details.

\subsection{Dynamics}

The relation between classical and quantum dynamics is given by the Egorov theorem. Let us first recall a weak version easily deduced from Heisenberg equation and the expansion (\ref{eq:expansion_product}). Take any time dependent Hamiltonian $f_t$. Denote by $\phi_t: M \to M$ the Hamiltonian flow generated by $f_t$ and by $U_{\hbar}(t) : \Hilb_{\hbar} \rightarrow \Hilb_{\hbar}$ the Schr\"{o}dinger evolution generated by the Hamiltonian $T_\hbar(f_t)$. Then for any smooth function $g$ on $M$, there exists a sequence $g_{\ell}$, $\ell \in \N^*$ of smooth functions such that for any $N$,
\begin{gather} \label{eq:egorov_weak}
 U_{\hbar} T_{\hbar} (g) U_{\hbar}^* = T_{\hbar} ( g \circ \phi^{-1} ) + \sum_{\ell = 1 }^{N} \hbar^{\ell } T_{\hbar} (g_{\ell} \circ \phi^{-1}) + \bigo ( \hbar^{N+1})
\end{gather}
where $U_{\hbar} = U_{\hbar}(1)$ and $\phi = \phi_1$. Furthermore, for any $\ell \in \N^*$, the support of $g_{\ell}$ is contained in the support of $g$.

In the sequel we will need a more precise version for $N=1$ with an explicit control of the remainder in terms of $f$ and $g$. This can be easily deduced from (\ref{eq:P2}) as follows. Put $g_t:= g \circ \phi_t^{-1}$ and note that $\dot{g}_t= \{f_t,g_t\}$.
Hence we have that
$$\frac{d}{d t}\Bigl( U_{\hbar}(t)^*T_\hbar(g_t)U_{\hbar}(t) \Bigr) = U_{\hbar}(t)^*\Bigl( \frac{i}{\hbar} [T_\hbar(f_t),T_\hbar(g_t)]+ T_\hbar(\{f_t,g_t\}\Bigr) U_{\hbar}(t) \;,$$
and therefore by (\ref{eq:P2})
$$ \Bigl\| \frac{d}{d t} U_{\hbar}(t)^*T_\hbar(g_t)U_{\hbar}(t) \Bigr\|_{op} \leq \beta    |f_t,g_t|_{1,3}  \ \hbar\;.$$
Integrating this from $0$ to $1$, conjugating by $U_{\hbar}$ and recalling that $U_{\hbar}$ is unitary,  we get
\begin{equation}\label{eq-Egorov}
\| T_\hbar(g \circ \phi^{-1}) -U_{\hbar}T_\hbar(g)U_{\hbar}^{*}\|_{op} \leq \beta \hbar \int_0^1 |f_t,g \circ \phi_t^{-1}|_ {1,3}dt \;.
\end{equation}


\section{When dislocation yields displacement}\label{sec-wdyd}

In the present section we prove Theorem \ref{thm-intro-main-1}.
Our strategy is as follows. For a pair of quantum states  we
introduce an operator norm cousin of fidelity $\Gamma_q$ defined in \eqref{eq-Gammaq-def}
below. On the one hand, an easy linear algebra argument (Proposition \ref{prop-fidel}) enables us to relate $\Gamma_q$ to fidelity. On the other hand, $\Gamma_q$ admits a classical counterpart $\Gamma_{cl}$, see \eqref{eq-Gamma-disloc} below, which, roughly speaking, enables one to think about ``dislocation of functions" instead of displacement of subsets, and hence is closely related to the notion of symplectic dislacement energy discussed above. Thus for proving Theorem \ref{thm-intro-main-1} it suffices to quantify the relation between $\Gamma_q$ and $\Gamma_{cl}$. This is done in Theorem \ref{thm-disloc-11}  below with the help of quantum-classical correspondence formalized within the framework of the Berezin quantization  $(\Hilb_{\hbar}, \hbar \in \Lambda)$ satisfying (\ref{eq:estime_diag}), (\ref{eq:P1}), (\ref{eq:P2}) and (\ref{eq:P3}).

\subsection{An operator norm cousin of fidelity }\label{subsec-fid}
For a pair $\theta,\sigma$ of positive non-vanishing Hermitian operators on a $d$-dimensional Hilbert space $\Hilb$ put
\begin{equation}\label{eq-Gammaq-def}
 \Gamma_q(\theta,\sigma) = \frac{\|\theta\sigma\|_{op}} {\|\theta\|_{op} \|\sigma\|_{op}} \in [0,1]\;,
 \end{equation}
where $q$ stands for quantum. When $\theta$ and $\sigma$ are quantum states, this quantity serves as yet another measure of overlap. Let us discuss it in more details.

Observe that $\Gamma_q(\theta,\sigma)=0$ if and only if  $\theta$ and $\sigma$ have orthogonal images,
i.e., the states do not overlap in the very strong sense.

At the same time, equality $\Gamma_q(\theta,\sigma)=1$ is less informative. It does not necessarily
yields $\sigma=\theta$. For instance, $\Gamma_q(\theta, d^{-1}\id) = 1$ for all $\theta \in \cS$.

Let us illustrate the introduced notion in the following two cases. First, assume that
$\theta$ and $\sigma$ commute.
Fix a common eigenbasis and look at the probability distributions $p=\{p_1,...,p_d\}$ and $q=\{q_1,...,q_d\}$ on $\{1,...,d\}$ formed by the eigenvalues of $\theta$ and $\sigma$, respectively.
Vanishing of $\Gamma_q$ is equivalent to the fact that $p$ and $q$ have disjoint supports. However
$\Gamma_q=1$ means only that $p$ and $q$ attain maximum at the same point $i \in \{1,...,d\}$.

Second, let $\theta$ and $\sigma$ be pure states defined by unit vectors
$\xi,\eta \in \Hilb$, that is we identify $\theta$ and $\sigma$
with the orthogonal projectors to   $\xi$ and $\eta$, respectively. A direct calculation shows that $$\Gamma_q(\xi,\eta) = \Phi(\xi,\eta) = |\langle \xi,\eta \rangle |\;.$$

\begin{prop}\label{prop-fidel}
\begin{equation}\label{eq-fidel}
\Gamma_q(\theta,\sigma) \leq \frac{\Phi(\theta,\sigma)}{\|\theta\|_{op}^{1/2}\|\sigma\|_{op}^{1/2}}
\end{equation}
for all $\theta,\sigma \in \cS( \Hilb )$.
\end{prop}

\begin{proof} First observe that
\begin{equation}
\label{eq-fidel-vsp-1}
\Phi(\theta,\sigma) = \text{trace}\sqrt{\sigma^{1/2}\theta\sigma^{1/2}} \geq
\sqrt{\text{trace}(\sigma^{1/2}\theta\sigma^{1/2})}\;.
\end{equation}
Since $\theta$ is non negative, $\|\theta\|_{op} \cdot \theta \geq \theta^2$, so
$\sigma^{1/2}\theta\sigma^{1/2} \geq \|\theta\|_{op}^{-1}\cdot \sigma^{1/2}\theta^2\sigma^{1/2}$.
Thus, $$\text{trace}(\sigma^{1/2}\theta\sigma^{1/2}) \geq \|\theta\|_{op}^{-1}\cdot \text{trace}(\sigma^{1/2}\theta^2\sigma^{1/2})= \|\theta\|_{op}^{-1} \cdot \text{trace}(\theta \sigma\theta) \; .$$
In the same way, $\si$ being non negative, $\op{trace} (\theta \sigma\theta) \geq \|\sigma\|_{op}^{-1}\cdot \text{trace}(\theta \sigma^2 \theta)$ and we obtain
\begin{xalignat*}{2}
\text{trace}(\sigma^{1/2}\theta\sigma^{1/2}) & \geq  \|\theta\|_{op}^{-1}\|\sigma\|_{op}^{-1}\cdot \text{trace}(\theta \sigma^2 \theta)\\
&  \geq \|\theta\|_{op}^{-1}\|\sigma\|_{op}^{-1}\cdot \|\theta \sigma^2 \theta\|_{op} =
\|\theta\|_{op}^{-1}\|\sigma\|_{op}^{-1}\cdot \|\theta\sigma\|_{op}^2
\end{xalignat*}
where we have used that $\theta \sigma ^ 2 \theta $ is non negative.
Combining this with \eqref{eq-fidel-vsp-1}, we get inequality \eqref{eq-fidel}.
\end{proof}

\subsection{Classical dislocation}
For a pair of smooth functions $g,h \in \Cl^\infty(M) \setminus \{0\}$ define the
{\it classical} overlap
\begin{equation} \label{eq-Gamma-disloc}
\Gamma_{cl}(g,h):= \frac{\|gh\|}{\|g\|\cdot\|h\|}\;.
\end{equation}
Consider a Hamiltonian flow $\phi_t$ on a symplectic manifold $(M,\omega)$
generated by a compactly supported Hamiltonian $f_t$. Denote by $\phi=\phi_1$ its time-one map.
We say that the flow {\it $a$-dislocates} a function $g \in \Cl^\infty(M) \setminus \{0\}$, $a \in [0,1)$,
if
$$\Gamma_{cl}(g, g \circ \phi^{-1}) \leq a\;.$$
Put $X:= \{|g| > \sqrt{a}\|g\|\}$. Clearly, $a$-dislocation yields $\phi(X) \cap X = \emptyset$.
In particular, the Hofer length of the path $\{\phi_t\}$ satisfies
\begin{equation}\label{eq-ell-cl}
\ell_{cl}(f) := \int_0^1\|f_t\|dt \geq e_M(g,\sqrt{a}) >0\;,
\end{equation}
where by definition $e_M(g,r)$ is Hofer's displacement energy of $\{|g|>r\|g\|\}$ in $M$.

\subsection{Comparing $\Gamma_{cl}$ and $\Gamma_{q}$}\label{subsec-clsemicldisloc}
Fix the following classical data:
\begin{itemize}
\item  a non-negative function $g \in \Cl^\infty(M)$ with $\max g=1$;
\item  a time-dependent Hamiltonian $f_{t}$ generating a Hamiltonian flow $\phi_{t}$ with $\phi_{1}=\phi$.
\end{itemize}

\noindent
Produce the following quantum data:
\begin{itemize}
\item The operator
\begin{equation}\label{eq-thetahbar}
\theta_\hbar := T_\hbar(g) \in \cL ( \Hilb_\hbar)\;.
\end{equation}
\item The Schr\"{o}dinger evolution $U_{\hbar}(t)$ on $\Hilb_\hbar$ generated by the Hamiltonian $F_{t,\hbar}:= T_\hbar(f_{t})$
with $U_{\hbar}(1)=U_\hbar$.
\end{itemize}
Introduce the following numerical characteristics of the classical data:
\begin{equation}\label{eq-b-defin}
\begin{split}
b(g,f) :=\max \Big{(}\alpha |g|_2,
\alpha |g \circ \phi^{-1}|_2, \alpha|g \cdot g\circ \phi^{-1}|_2, \\ \beta\int_0^1 |f_{t},g \circ \phi_{t}^{-1}|_{1,3}dt, \gamma |g,g \circ \phi^{-1}|_2
\Big{)}
\end{split}
\end{equation}
and
\begin{equation}\label{eq-c-defin}
c(f) := \alpha\int_0^1 |f_{t}|_2 dt\;,
\end{equation}
where $\alpha,\beta,\gamma$ are the constants in the properties (\ref{eq:P1}),(\ref{eq:P2}),(\ref{eq:P3}). The derivatives entering in the definition of $c(f)$ and $b (g,f)$ will play an important role in the proof of Theorem \ref{thm-disloc-1}, which is the generalization of Theorem \ref{thm-intro-main-1} to small scale.

We shall assume that $b(g,f) \hbar< 1$. Then by (\ref{eq:P1}),  $T_\hbar(g) \neq 0 $.

Put $$\Gamma_{q,\hbar}:= \Gamma_{q}(\theta_\hbar, U_\hbar\theta_\hbar U_\hbar^{-1}),\;\; \Gamma_{cl}:= \Gamma_{cl}(g,g\circ \phi^{-1})\;,$$
and
$$\ell_{cl}= \ell_{cl}(f),\;\;\ell_{q,\hbar}= \ell_q(F _{\hbar})\;.$$
Write $b =  b ( g,f)$ and $c = c(f)$.

\medskip
\noindent
\begin{thm}\label{thm-disloc-11} When  $b \hbar< 1$, we have that
\begin{equation}\label{eq-disloc-main}
\Gamma_{cl} -3b \hbar \leq \Gamma_{q,\hbar} \leq \frac{\Gamma_{cl}+ 2b \hbar}{(1-b \hbar)^2} .
\end{equation}
Furthermore,
\begin{equation}\label{eq-disloc-ell}
\ell_{cl}-c \hbar \leq \ell_{q,\hbar} \leq \ell_{cl}\;.
\end{equation}
\end{thm}

\medskip
\noindent
\begin{cor}\label{cor-main-disloc} Assume that $f_{t}$ and $g$ are compactly supported in an open subset $W \subset M$. Then the time-one-map of $f_{t}$ displaces the set
$$\bigl\{g > \sqrt{\Gamma_{q,\hbar}+ 3b\hbar}\bigr\}$$
and
\begin{equation}
\label{eq-main-disloc}
\ell_{q,\hbar} \geq e_W \bigl( g, \sqrt{\Gamma_{q,\hbar}+ 3b\hbar} \bigr) -c\hbar\;
\end{equation}
provided $b\hbar < 1$.
\end{cor}

\medskip
\noindent
\begin{proof} By Theorem \ref{thm-disloc-11}
$$\Gamma_{cl} \leq \Gamma_{q,\hbar} + 3b\hbar\;,$$
which yields the statement about displacement. Furthermore,
$$\ell_{q,\hbar} \geq \ell_{cl,\hbar} - c\hbar\;.$$
By \eqref{eq-ell-cl}
$$\ell_{cl} \geq e_W \bigl( g, \sqrt{\Gamma_{cl}} \bigr) \geq e_W \bigl( g, \sqrt{\Gamma_{q,\hbar}+ 3b\hbar} \bigr) \;.$$
Combining these inequalities, we get \eqref{eq-main-disloc}.
\end{proof}

\medskip
\noindent
{\bf Proof of Theorem \ref{thm-disloc-11}:}
By Egorov theorem \eqref{eq-Egorov}, (\ref{eq:P1}) and (\ref{eq:P3}) we have that
\begin{gather*}
\| T_\hbar(g \circ \phi^{-1}) -U_\hbar T_\hbar(g)U_\hbar^{-1} \|_{op} \leq  b\hbar\;, \\
\|T_\hbar(g)T_\hbar(g \circ \phi^{-1}) - T_\hbar(g \cdot g \circ \phi^{-1}) \|_{op} \leq b\hbar\;, \\
 1 - b\hbar \leq \|T_\hbar(g)\|_{op} \leq 1 \;,\\
1- b\hbar \leq \|T_\hbar(g \circ \phi^{-1})\|_{op} \leq 1 \;, \\
\|g \cdot g\circ \phi^{-1}\|  - b\hbar \leq \|T_\hbar(g \cdot g\circ \phi^{-1})\|_{op} \leq \|g \cdot g\circ \phi^{-1}\|\;,
\end{gather*}
and
\begin{equation}\label{eq-vsp-rk18}
\|f_{t}\|-\alpha|f_t|_2\hbar \leq \|F_{t, \hbar}\|_{op} \leq \|f_{t}\|\;.
\end{equation}
We deduce that
\begin{gather*}
\Gamma_{q,\hbar} = \frac{\|T_\hbar(g)U_\hbar T_\hbar(g)U_\hbar^{-1}\|_{op}}{\|T_\hbar(g)\|_{op} \cdot \|T_\hbar(g\circ \phi^{-1})\|_{op}} \leq \frac{\|T_\hbar(g)T_\hbar(g \circ \phi^{-1})\|_{op}+b\hbar}{(1-b\hbar)^2} \\
\leq \frac{\|T_\hbar(g \cdot g\circ \phi^{-1})\|_{op} + 2 b \hbar}{(1-b\hbar)^2} \leq \frac{\Gamma_{cl}+ 2b\hbar}{(1-b\hbar)^2}
\end{gather*}
and
\begin{gather*}
\Gamma_{q,\hbar} = \frac{\|T_\hbar(g)U_\hbar T_\hbar(g)U_\hbar^{-1}\|_{op}}{\|T_\hbar(g)\|_{op} \cdot \|T_\hbar(g\circ \phi^{-1})\|_{op}} \geq \|T_\hbar(g)T_\hbar(g \circ \phi^{-1})\|_{op}- b\hbar \\
 \geq \|T_\hbar(g \cdot g\circ \phi^{-1})\|_{op}-2b\hbar \geq
\Gamma_{cl}- 3b\hbar\;,
\end{gather*}
which proves \eqref{eq-disloc-main}.

Integrating \eqref{eq-vsp-rk18} over the time interval $[0,1]$ and applying  (\ref{eq:P1})
we immediately deduce inequality \eqref{eq-disloc-ell}.
\qed

\subsection{Proof of Theorem \ref{thm-intro-main-1}} \label{sec:proof-theorem-1.2ii}

We need the following control on the constant $b=b(g,f)$.

\begin{lemma} \label{lem:constant_control}
The constant $b(g,f)$ remains bounded when $\sup _{t \in [0,1] } |f_t|_4$ and $|g|_3$ remain bounded.
\end{lemma}

\medskip
\noindent{\bf Sketch of the proof:} The Hamiltonian vector field of $f_t$ is bounded with $3$ derivatives.
This implies, by the standard ODE arguments (the higher order variational equation combined with the Gr\"{o}nwall's lemma, cf. \cite[Section 2.4]{T12}), that the Hamiltonian flow $\phi_t$ is also bounded with $3$ derivatives. Combined with the fact that $g$ is bounded with $3$ derivatives,
this readily yields the statement of the lemma.
\qed

\medskip

For $u \in \Cl^ 3 (M)$ with $\int_M ud\mu=1$ define the quantum state $\theta_{\hbar} := Q_{\hbar} ( u d\mu)$.  Set $v_{\hbar} = u  / R_{\hbar}$ so that $\theta_{\hbar} = T_{\hbar} ( v_{\hbar})$. Let $g_{\hbar} = v_{\hbar} / \| v_{\hbar} \| $.
Let $(f_{t,\hbar})$ be a family of Hamiltonians which are uniformly (in $\hbar$) bounded in $\Cl ^3 (M)$.

We will apply the result of the previous section to the $\hbar$ dependent data $g_{\hbar}$ and $(f_{t,\hbar})$. Set $b_\hbar = b  ( g_{\hbar}, f_{t, \hbar})$ and $c_{\hbar}= c ( f_{t,\hbar})$. Then
\begin{enumerate}
\item  It follows from (\ref{eq:estime_diag}) that $| g_{\hbar} |_3$ is bounded independently of $\hbar$. Thus by Lemma \ref{lem:constant_control}, there exists $c >0$ such that the constants defined in (\ref{eq-b-defin}) and (\ref{eq-c-defin}) satisfy for any $\hbar \in \Lambda$, $b_{\hbar} \leqslant c$ and $c_{\hbar} \leqslant c$.
\item
By (\ref{eq:P1}), there exists $c'>0$ such that when $\hbar$ is sufficiently small, we have $\| \theta_{\hbar}  \| \geqslant \hbar^n /c'$.
\end{enumerate}
Assume that the Schr\"odinger flow generated by $T_{\hbar} ( f_{t, \hbar})$ $a_\hbar$-dislocate $\theta_\hbar $. Proposition \ref{prop-fidel} yields $\Ga_{q, \hbar} \leqslant c' a_\hbar \hbar^{-n}$. By Corollary \ref{cor-main-disloc}, the time-one-map of $f_{t,\hbar}$ displaces the set $\{ g_{\hbar}  > A \}$ with
\begin{gather} \label{eq:estimA}
 A = \sqrt{\Gamma_{q,\hbar}+ 3b_\hbar\hbar}  \leqslant \sqrt{c'a_{\hbar} \hbar^{-n} + 3 c \hbar } \; .
\end{gather}
Let $\lambda >0$. When $\hbar$ and $a_\hbar \hbar^{-n}$ are sufficiently small, $ A \leqslant  \lambda /(2 \|u \|)$, so that $\{ g_{\hbar} > \lambda /(2\|u\|) \}$ is contained in $\{ g_{\hbar} > A \}$. Furthermore, by (\ref{eq:estime_diag}), when $\hbar$ is sufficiently small $g_{\hbar} \geqslant u/(2\|u \|)$, so that $\{ u > \lambda \}$ is contained in $\{ g_{\hbar} > \lambda/ (2\|u\|) \}$. Thus $\{ u > \lambda \}$  is displaced by the time-one-map of $f_{t,\hbar}$. This implies by (\ref{eq-disloc-ell}) that
\begin{equation}
\label{eq-disloc-unitscale}
\ell_q(F_\hbar) \geqslant \ell_{cl} - c \hbar \geqslant  e_M( \{ u > \lambda\}) - c \hbar\;.
\end{equation}
This completes the proof of Theorem \ref{thm-intro-main-1}.
\qed

The above proof highlights our assumption on the $\Cl^3$-smoothness of the classical density $u$.
We use it in order to apply sharp  remainder estimates for the semi-classical remainders.


\section{Microsupport}\label{sec-MS}

For a continuous section $s$ of a vector bundle $E \rightarrow X$, note that the support of $s$ is $\bigcap \{ x /  f (x) = 0 \}$ where the intersection is taken over all $f \in \Cl^0 (M)$ such that $f s=0$. Replacing the zero sets of functions by the characteristic sets of pseudodifferential operators, H\"ormander defined the wave front of a distribution \cite{Ho_acta1}. Using Toeplitz operators, we can define similarly a microsupport of a family $(\Psi_\hbar \in \Hilb_\hbar, \; \hbar \in \Lambda)$ as was done in \cite{oim1} for K\"ahler manifold.
Here we will consider more generally the microsupport of a mixed state. Consider a family $(\Hilb_{\hbar} , \hbar \in \Lambda)$ of Berezin quantizations satisfying  (\ref{eq:estime_diag}), (\ref{eq:estime_outside_diag}) and (\ref{eq:expansion_product}).

\subsection{Definitions and examples}

Let $\theta$ be a family $( \theta_\hbar \in \cS (\Hilb_\hbar), \; \hbar \in \Lambda )$.
By definition, the {\em microsupport} of $\theta$ is the subset $\op{MS} ( \theta)$  of $M$ such that:  $x \notin \op{MS} ( \theta)$ if and only if there exists $f \in \Ci (M)$  such that $f(x) \neq 0$ and
$$T_\hbar (f) \theta_\hbar = \bigo ( \hbar ^{\infty}).$$
Here the $\bigo(\hbar^{\infty})$ is for the operator     norm. Since the dimension of $\Hilb_\hbar$ is a $\bigo ( \hbar^{-n})$ we could equivalently work with trace norm or Hilbert-Schmidt norm.

 The microsupport is certainly closed. It can not be empty. Indeed, if it were empty, then by the following lemma, $\theta_{\hbar} = \bigo ( \hbar ^{\infty})$ contradicting the fact $\| \theta_\hbar \|_{tr} = 1$.

\begin{lemma} \label{lem:cut}
For any family $\theta$ and function $u \in \Ci (M )$ such that the microsupport of $\theta$ and  the support of $u$ are disjoint,  $T_\hbar (u) \theta_\hbar = \bigo ( \hbar^{\infty})$.
\end{lemma}

\begin{proof}
 Let $x_0 \notin \op{MS}( \theta)$. We will prove that for any function $u$ supported in a sufficiently small neighborhood of $x_0$, we have $T_\hbar (u) \theta_\hbar = \bigo ( \hbar^{\infty})$. It is then easy to deduce the lemma with a partition of unity argument.

Let $ f\in \Ci (M)$ be such that $f(x_0) \neq 0$ and $T_\hbar (f) \theta_\hbar = \bigo ( \hbar^{\infty})$. Choose a neighborhood $U$ of $x_0$ such that $| f | \geqslant c >0$ on $U$. Let $\chi \in \Ci ( \R, \R)$ be equal to $0$ on $]-\infty, 1/2]$ and to $1$ on $[1, \infty[$. The function $g_0 =  \chi  ( |f|/ c)  f^{-1}  $ is smooth, supported in $\{ |f | \geqslant c/2\}$ and  $g_0 f =1$ on $U$. Using (\ref{eq:expansion_product}), we construct by successive approximations a sequence $(g_{\ell})$ of $\Ci (M)$ such that for any $N \in \N$,
$$ T_\hbar( g_0 + \hbar g_1 + \ldots + \hbar^N g_N ) T_\hbar(f) = T_\hbar( g_0 f ) + \bigo ( \hbar^{N +1} )$$
Since $T_\hbar( f) \theta_\hbar = \bigo ( \hbar^{\infty})$, we have $T_{\hbar} (g_0 f) \theta_\hbar = \bigo ( \hbar^{\infty})$. For any function $ u \in \Ci (M)$ supported in $U$, $u$ and $g_0 f-1 $ have disjoint supports, so that $T_\hbar ( u) = T_\hbar ( u) T_\hbar ( g_0 f) + \bigo ( \hbar^{\infty})$, hence  $T_\hbar ( u ) \theta_\hbar = \bigo ( \hbar^{\infty})$.
\end{proof}

\begin{exam}
\begin{enumerate}
\item
The microsupport of the normalized coherent state at $x$ is $\{ x\}$. Indeed, it follows from  (\ref{eq:estime_outside_diag}) that the microsupport is contained in $\{ x \}$ and the microsupport can not be empty.
\item The microsupport of a Lagrangian state (see Section \ref{subsec-lasd} below) whose symbol does not vanish is the associated Lagrangian manifold. This follows readily from the characterization of the microsupport  given in Proposition \ref{prop:pure_state}.
\end{enumerate}
\end{exam}

\begin{prop} \label{prop:pure_state}
\noindent
\begin{enumerate}
\item For any pure state $\psi =    ( \psi_\hbar \in \Hilb_\hbar, \; \hbar \in \Lambda)$, for any $x \in M$,  $ x \notin \op{MS} ( \psi)$ if and only if there exists a neighborhood $U$ of $x$, such that for any $N$, $| \psi_\hbar (y) | = \bigo  (\hbar^{N})$ uniformly in $y \in U$.
\item  Any non empty closed set is the microsupport of some pure state.
\end{enumerate}
\end{prop}

The characterization of the first assertion was actually the definition of \cite{oim1}. The proof of the second assertion is surprisingly much simpler than the proof of the similar result for wave front (cf. Theorem 8.1.4 in \cite{Ho1}).

\begin{proof}
If $\psi_\hbar = \bigo ( \hbar^{\infty})$ on a neighborhood $U$ of $x$, then for any $f$ supported in $U$, $\| f \psi_\hbar \| = \bigo ( \hbar^{\infty})$. Since $\| T_{\hbar} (f) \psi_{\hbar} \| \leqslant \| f \psi_{\hbar} \|$, it follows that  $x \notin \op{MS} ( \psi)$. Conversely, if $x \notin \op{MS}(\psi)$, by Lemma \ref{lem:cut}, for any $f$ supported in a sufficiently small neighborhood of $x$, $T_\hbar(f) \psi_\hbar = \bigo ( \hbar^{\infty})$. Choose $f$ so that it is equal to $1$ on a neighborhood of $x$. Then for any $y $ sufficiently close to $x$, by (\ref{eq:estime_outside_diag}), we have $|e_{y, \hbar} (z) | = \bigo( \hbar^{\infty})$ when $f(z) \neq 1$. So $f e_{y,\hbar} = e_{y,\hbar} + \bigo ( \hbar^{\infty})$ uniformly with respect to $y$. Consequently,
$$   \psi_\hbar (y) =  \langle \psi_\hbar , e_{y,\hbar} \rangle = \langle \psi_\hbar , f e_{y,\hbar} \rangle + \bigo ( \hbar^{\infty}) = \langle f \psi_\hbar , e_{y,\hbar} \rangle + \bigo ( \hbar^{\infty}) = \bigo ( \hbar^{\infty})$$ which proves the first assertion.

As in the proof of Theorem 8.1.4 of \cite{Ho1}, note that for any non empty closed set $C$ of $M$, there exists a countable set $D$ of $M$ whose set of limit points is exactly $C$. Let $(x_k, \; k \in \N^* )$ be an enumeration of $D$. Choose a decreasing sequence $(\hbar_k )$ of $\Lambda$ converging to $0$. For any $\hbar \in \Lambda$, set  $y_{\hbar} = x_k$ if $\hbar \in ] \hbar_{k+1}, \hbar_k] $.
Define $\Psi_{\hbar} = e_{y_\hbar, \hbar} / \|e_{y_\hbar, \hbar} \| $. Using the first assertion of the proposition and estimates (\ref{eq:estime_diag}), (\ref{eq:estime_outside_diag}),  we prove that $\op{MS} ( \Psi) = C$.
\end{proof}

\subsection{Microsupport and classical measures}\label{subsec-mscm}

We will now describe the microsupport of $\theta$ in terms of the corresponding semiclassical measure $\nu_{\hbar}$ given by $\int_M f d \nu_{\hbar}  = \op{trace} (T_{\hbar} (f) \theta_\hbar) $.

\begin{prop} \label{prop:micr-class-meas}
 $x \notin \op{MS} ( \theta)$ if and only if $x$ has a neighborhood $U$ such that $\nu_\hbar (U) = \bigo ( \hbar^{\infty})$
\end{prop}
The converse of this equivalence seems to be new, even for pure states.
\begin{proof} If $x$ is not in the microsupport, choose a non negative smooth function $f$ equal to 1 on a neighborhood $U$ of $x$ and such that the support of $f$ is disjoint of $\op{MS} ( \theta)$. Then $T_\hbar ( f) \theta_\hbar = \bigo ( \hbar^{\infty})$ and
$$\nu_\hbar (U) \leqslant \int f d\nu_\hbar \leqslant \| T_\hbar ( f) \theta_\hbar \|_{tr}.$$
So $\nu_\hbar ( U) = \bigo ( \hbar^{\infty})$.

Conversely, if $\nu_\hbar (U) = \bigo ( \hbar^{\infty})$ and $f$ is smooth, real valued  and supported in $U$, then $\op{trace} ( T_\hbar (f^2) \theta_\hbar ) = \int f^2 d\nu_\hbar = \bigo ( \hbar^{\infty})$. To complete the proof, we will show that this implies $T_\hbar (f) \theta_\hbar = \bigo ( \hbar^{\infty})$.

Since  $\| T_{\hbar} (f) s \| \leqslant \| f s\|$,
$\langle T_{\hbar} (f^2) s, s \rangle = \langle f^2 s, s\rangle = \| fs \|^2 \geqslant \| T_{\hbar}(f) s \| ^2 = \langle T_{\hbar} (f)^2 s, s \rangle$. Hence we have proved that
\begin{gather}\label{eq:ineq1}
T_{\hbar} ( f^2 ) \geqslant T_{\hbar} (f) ^2
\end{gather}
$\theta_\hbar$ being non negative and having trace 1, we have
\begin{gather} \label{eq:ineq2}
\theta_{\hbar} \geqslant \theta_\hbar^2
\end{gather}
We also note that if $A \geqslant B$ and $C \geqslant 0$, then
\begin{gather} \label{eq:ineq3}
 \op{trace} (AC) = \op{trace} ( C^{1/2} A C^{1/2}) \geqslant \op{trace} ( C^{1/2} B C^{1/2}) = \op{trace} (BC) \; .
\end{gather}
Denote $D_{\hbar} = T_{\hbar} (f) \theta_{\hbar}$. Combining (\ref{eq:ineq3}) with (\ref{eq:ineq1}) and (\ref{eq:ineq2}), we get
\begin{xalignat*}{2}
 \op{trace} ( T_{\hbar} (f^2 ) \theta_\hbar ) & \geqslant \op{trace} ( T_{\hbar}(f)^2 \theta_{\hbar} )
 \geqslant \op{trace} ( T_{\hbar} (f)^2 \theta_\hbar^2) \\
& = \op{trace} ( D_{\hbar} D_{\hbar}^* ) \geqslant \| D_{\hbar} \|_{op}^2
\end{xalignat*}
Since $ \op{trace} (T_{\hbar} (f^2 ) \theta_\hbar) = \bigo ( \hbar ^{\infty})$, we conclude that $ D_{\hbar}  = \bigo (\hbar^{ \infty})$.
\end{proof}

\begin{cor} \label{cor}
If $\nu_\hbar$ converges weakly to some measure $\nu$, in the sense that for any $f \in \Ci (M)$, $ \int f d \nu_\hbar \rightarrow \int f d\nu$, then the support of $\nu$ is contained in the microsupport of $\theta$.
\end{cor}

For pure states in the Bargmann space, this corollary is proposition 29 in \cite{Combescure_Robert}. The microsupport is called the frequency set in this book.

From Proposition \ref{prop:micr-class-meas}, we can also deduce that the microsupport of the quantum state $Q_\hbar ( \tau) = \int P_{\hbar,x} d\tau (x) $ is the support of $\tau$.  We will consider more generally a classical state $\tau$ depending on $\hbar$.

\begin{prop} \label{prop:quantum_state_microsupport}
Let $( \tau_{\hbar}, \; \hbar \in \Lambda)$ be a family of Borel probability measures and $\theta = (Q_\hbar ( \tau_{\hbar}))$. Then $x \notin \op{MS} ( \theta) $ if and only if $x$ has a neighborhood $U$ such that $\tau_{\hbar} (U)= \bigo ( \hbar^{\infty})$.
\end{prop}

\begin{proof}
Let $\nu_\hbar$ be the semiclassical measure of $Q_{\hbar} ( \tau_{\hbar})$. By Proposition \ref{prop:micr-class-meas}, it suffices to show that for any open sets $U$, $V$ such that the closure $\op{cl}(V)$ of $V$ is contained in $U$, $\tau_{\hbar} (U) = \bigo( \hbar^{\infty})$ implies that $\nu_{\hbar} (V) = \bigo ( \hbar ^{\infty})$ and $\nu_{\hbar} (U)  = \bigo( \hbar^{\infty})$ implies that $\tau_{\hbar} (V) = \bigo ( \hbar ^{\infty})$

Choose a non negative function $f$ such that $ V \subset \{ f =1 \}$ and $\op{supp} f \subset U$. Recall that $ \langle \nu_{\hbar} , f \rangle = \langle \tau_{\hbar},  \mathcal{B}_{\hbar} (f) \rangle$ where $\langle \cdot, \cdot \rangle$ stands for the pairing between Borel measures and functions.
We will use the following properties of the Berezin transform $\mathcal{B}_\hbar$. First, by definition (\ref{eq:def_Berezin_transform}), we certainly have for any $g$
\begin{gather} \label{eq:est_berezin_transf}
\| \mathcal{B}_{\hbar} (g) \| \leqslant \|g \|
\end{gather}
Rewrite
$$ \mathcal{B}_{\hbar} ( g) (x) = \int_M g(y) \frac{ |\langle e_x, e_y \rangle |^2}{\| e_x \|^2} d \mu (y).$$
By (\ref{eq:estime_diag}) and (\ref{eq:estime_outside_diag}), for any disjoint closed sets $C_1$, $C_2$,
\begin{gather} \label{eq:estim_kern_Berezin}
 \int_{C_2}  \frac{ |\langle e_x, e_y \rangle |^2}{\| e_x \|^2 } d \mu (y) = \bigo ( \hbar^{\infty}) \qquad \forall x \in C_1
\end{gather}
with a uniform $\bigo$. We deduce the following estimates
\begin{gather} \label{eq:est_berezin_transf_2}
\mathcal{B}_{\hbar} (f)  = \bigo(\hbar^{\infty}) \qquad  \text{ uniformly on } M \setminus U \\
\label{eq:est_berezin_transf_3}
\mathcal{B}_{\hbar}(f) = 1 + \bigo ( \hbar^{\infty} ) \qquad \text{ uniformly on $V$}.
\end{gather}
To prove (\ref{eq:est_berezin_transf_2}), set $C_1 = \op{supp} f$ and $C_2 = M \setminus U$ in (\ref{eq:estim_kern_Berezin}) and use that the support of $f$ is contained in $U$. To prove (\ref{eq:est_berezin_transf_3}), choose an open set $W$ such that $\op{cl} (V) \subset W \subset \{ f = 1\}$. By (\ref{eq:estim_kern_Berezin}) with $C_1 = \op{cl} (V)$ and $C_2 = M \setminus W$, we get that
$$\mathcal{B}_{\hbar}( f) (x) = \int_{W}  |\langle e_x, e_y \rangle |^2/\| e_x \|^2  d \mu (y) + \bigo(\hbar^{\infty}). $$
We can also apply this last equality to $f=1$ by choosing $U = M$. Since the right-hand side does not depend on $f$, we get $\mathcal{B}_{\hbar}(f)(x) = \mathcal{B}_{\hbar} (1) (x) + \bigo ( \hbar^{{\infty}}) $, for any $x \in \op{cl} (V)$. Since $\mathcal{B}_{\hbar}(1) =1$, this proves (\ref{eq:est_berezin_transf_3}).

Now, assume that $\tau_{\hbar} ( U) = \bigo ( \hbar^{ \infty})$. Then (\ref{eq:est_berezin_transf}) and (\ref{eq:est_berezin_transf_2}) yield that $\langle \tau_{\hbar}, \mathcal{B}_{\hbar}(f) \rangle = \bigo ( \hbar^{\infty})$. Since $ \nu_{\hbar} (V) \leqslant \langle \nu_{\hbar} , f \rangle = \langle \tau_{\hbar}, \mathcal{B}_{\hbar}(f) \rangle$, we get that  $\nu_{\hbar} (V) = \bigo ( \hbar^{\infty})$.

Conversely, assume that $\nu_{\hbar} (U) = \bigo ( \hbar^{ \infty})$ so that $\langle \nu_{\hbar}, f \rangle = \bigo  (\hbar^{\infty})$. Then $\mathcal{B}_{\hbar}(f)$ being non negative, we have by (\ref{eq:est_berezin_transf_3}),
$$ \tau_{\hbar} (V) \leqslant \langle \tau_{\hbar}, \mathcal{B}_{\hbar}(f) \rangle + \bigo ( \hbar^{ \infty}) = \langle \nu_{\hbar} , f \rangle + \bigo( \hbar^{ \infty}) = \bigo ( \hbar^{\infty}) $$
which completes the proof.
\end{proof}

\subsection{Dislocation and propagation}\label{subsec-dp}
Here we establish basic facts about dislocation and propagation of microsupports which
were used for the proofs of the results of Section \ref{subsec-dyd} of the introduction.

\begin{prop} \label{prop:product_microsupport}
For any two states $\theta$, $\sigma$ with disjoint microsupports, we have $\theta_\hbar \sigma_\hbar = \bigo(\hbar^{\infty})$ and $\Phi(\theta,\sigma) = \bigo(\hbar^{\infty})$.
\end{prop}

\begin{proof} Let $f, g$ be functions with disjoint supports such that $f = 1$ (resp. $g=1$) on a neighborhood of $\op{MS}( \theta)$ (resp. $\op{MS} ( \sigma)$). By Lemma \ref{lem:cut}, $T_\hbar(f) \theta_\hbar = \theta_\hbar +  \bigo ( \hbar^{\infty} )$ and $T_\hbar (g) \sigma_\hbar = \sigma_\hbar + \bigo(\hbar^{\infty})$. Multiplying the adjoint of the first equality by the second one, we get
$$ \theta_\hbar \sigma_\hbar = \theta_\hbar T_\hbar ( \con{f})T_\hbar (g) \sigma_\hbar + \bigo ( \hbar^{\infty}) = \bigo ( \hbar^{\infty})$$ because $ T_\hbar ( \con{f})T_\hbar (g) = \bigo ( \hbar^{\infty})$ by (\ref{eq:expansion_product}).

Next, we use the  following elementary upper bound for fidelity.
For any quantum states $\theta$ and $\sigma$, put $A:= \sqrt{\theta}\sigma\sqrt{\theta}$. Then
\begin{xalignat}{2} \label{eq-fid-upbd}
\begin{split}
\Phi(\theta,\sigma)  = \text{trace}\sqrt{A}  & \leq
\sqrt{\dim \Hilb_\hbar}\cdot  \sqrt{\text{trace}(A)} \\ &  =  \sqrt{\dim \Hilb_\hbar}\cdot \sqrt{\text{trace}(\theta\sigma)} \leq \dim \Hilb_\hbar \cdot \sqrt{||\theta\sigma||_{op}}\;.
\end{split}
\end{xalignat}
Since $\dim \Hilb_\hbar = \bigo(\hbar^{-n})$, we conclude that $\Phi(\theta,\sigma) = \bigo(\hbar^{\infty})$.
\end{proof}

It is a well-known consequence of Egorov Theorem that the Schr\"odinger equation propagates the wave front of any distribution by the Hamiltonian flow of the corresponding classical Hamiltonian. The same arguments applies here. So consider a Hamiltonian $f_t, t \in [0,1]$. Let $\phi$ be the corresponding Hamiltonian diffeomorphism and $U_\hbar$ be the quantum propagator at time 1 generated by $T_{\hbar} ( f_t)$.  From (\ref{eq:egorov_weak}) and Lemma \ref{lem:cut}, we deduce the following result.

\begin{prop}\label{prop:propagation_microsupport}
For any state $\theta $, we have $\op{MS} ( U \theta U^* ) = \phi ( \op{MS} ( \theta))$.
\end{prop}

\noindent Intuitively speaking, this shows that the microsupport of a propagating quantum state evolves according to
the underlying Hamiltonian flow.

\begin{rem}\label{rem-hbardepham}{\rm
In Section \ref{sec-ddsc}, we will need the following variation of this proposition. Denote by $U_{\hbar}(t)$ the quantum propagator at time $t$ generated by $T_{\hbar} ( f_t)$. Choose a family $(t_{\hbar}, \hbar \in \Lambda)$ converging to $0$ as $\hbar \rightarrow 0$. Then the microsupport of $(U_{\hbar}( t_{\hbar}) \theta_\hbar  U_{\hbar}(t_{\hbar})^*)$ coincide with the microsupport of $\theta$. This result holds more generally for a quantum Hamiltonian of the form $T_{\hbar} ( f_t ( \cdot , \hbar))$ where $f_{t} ( \cdot, \hbar)$ is a familly of time dependent Hamiltonian admitting an asymptotic expansion of the form $ f_{t,0} + \hbar f_{t,1} + \hbar^2 f_{t,2} + \ldots $ whose coefficients $f_{t,\ell}$ are functions in $\Ci (M, \R)$ depending smoothly on $t$. To prove this, one has to extend equation (\ref{eq:egorov_weak}) to time $t \in [0,1]$ and one  uses again Lemma \ref{lem:cut}.}
\end{rem}

\section{Proof of Theorem \ref{thm-main-31}}\label{subsec-thm-main-31-pf}
We start with the following folkloric result.

\medskip
\noindent
\begin{lemma}\label{lem-folk} Let $M$ be a closed symplectic manifold, and
$W \subset M$ be an open subset with $\mu(W) < \mu(M)/2$. Then for every $\epsilon >0$, $\delta >0$
there exists a time-dependent Hamiltonian $g_t$, $t \in [0,1]$ with $\ell_{cl} \leq \delta$ such that
its time-one-map $\psi$  ``displaces $W$ up to measure $\epsilon$":
$\mu(W  \cap\psi(W))\leq \epsilon$.
\end{lemma}

\medskip
\noindent
We shall need a number of auxiliary statements.
By a {\it ball} in $M$ we mean an embedded open ball with smooth boundary.
Choose any Riemannian metric on $M$ with the distance function $d$.

\begin{lemma}\label{lem-1-mos} Let $Z \subset M$ be an open subset. For every $c>0,\epsilon >0$ there exists a collection $X_1,...,X_N \subset Z$ of balls of diameter $< c$ with pairwise disjoint closures such that $$\mu (Z \setminus \sqcup_j X_j) \leq \epsilon\;.$$
\end{lemma}
\begin{proof}
Fix a sufficiently small $c > 0$, and consider a triangulation of $M$ by smooth simplices (i.e., images of the standard simplex under a diffeomorphism) of diameter $< c$. Slightly shrinking $Z$ we can assume without loss of generality that $\partial Z$ is smooth. Indeed, $Z$ admits a smooth proper bounded from below function
\cite[Proposition 2.28]{Lee}, say $u: Z \to \R$. By Sard's lemma, there exists a sequence of regular values of $u$, $c_i \to +\infty$. The domains $\{u < c_i\}$ possess smooth boundaries and exhaust $Z$, so we replace
$Z$ by one of them.

Pick all the simplices lying strictly inside $Z$. Their union, say $Q$, covers the complement in $Z$ to the collar neighborhood of $\partial Z$
of width $c$ and hence $\mu(Z \setminus Q) = o(1)$ as $c \to 0$. Next, for every simplex $\Delta_i \subset Q$, $i=1,\dots,N$ remove a tiny collar neighbourhood of its boundary.
We get a topological ball $X_i$ of diameter $< c$ and with $\mu(\Delta_i \setminus X_i) < c/N$.
Thus $$\mu (Z \setminus \sqcup X_i) < c/N \cdot N + o(1) = o(1)\;,$$
and $\text{diam} (X_i) = o(1)$, so we are done by choosing $c$ sufficiently small.
\end{proof}

\medskip
\noindent
A Hamiltonian
diffeomorphism is the time-one map of the flow of a (time-dependent) Hamiltonian function.
All Hamiltonians below are assumed to be $\Ci$-smooth.

\begin{lemma}\label{lem-2-mos}
Let $\phi$ be any Hamiltonian diffeomorphism of $M$, and $Z \subset M$
be an open subset. Then for all $\epsilon>0,\delta>0$ there exists a Hamiltonian
$g_t$ with $\ell_{cl} \leq \delta$ and a subset $U \subset Z$ so that $\mu ( Z \setminus U )< \epsilon$
and the time-one-map $\psi$ of $g_t$ coincides with $\phi$ on $U$.
\end{lemma}

\begin{proof}
Let $f_t$ be a Hamiltonian generating a flow $\phi_t$ with $\phi_1=\phi$. Choose $L>0$ exceeding the Lipschitz constants of $f_t$ and $\phi_t$ for all $t$.
Let $X_1,...,X_N \subset Z$ be balls of diameter $< \delta/(2L^2)$ with pairwise disjoint closures such that
$$\mu (Z \setminus \sqcup_j X_j) \leq \epsilon\;,$$ see Lemma \ref{lem-1-mos}.
Pick a point $x_j \in X_j$ and observe that
$$\max_{X_j} |f_t(\phi_t(x)) - f_t(\phi_t(x_j))| \leq L^2 \cdot \frac{\delta}{2L^2} = \delta/2\;.$$
For every $t$, let $g^{(j)}_t$ be the cut-off of the function $f_t(x) - f_t(\phi_t(x_j))$ such that
$g^{(j)}_t = f_t(x) - f_t(\phi_t(x_j))$ in a tiny neighbourhood of $\phi_t(X_j)$,
$g^{(j)}_t$ vanishes outside a slightly bigger neighbourhood and the supports of $g^{(j)}_t$'s
are pair-wise disjoint for different $j$'s. One can always perform the cut-off so that
$g^{(j)}_t$ is smooth in $t$ and $\|g^{(j)}_t\| < \delta$. Observe that the Hamiltonian
flow $\psi^{(j)}_t$ generated by $g^{(j)}_t$ coincides with $\phi_t$ on $X_j$, and
the supports of $g^{(j)}_t$ for different $j$'s are disjoint. Thus, the product
$\psi_t:= \prod_j \psi^{(j)}_t$ is generated by the Hamiltonian
$g_t = \sum_j g^{(j)}_t$
and $$\|g_t\| \leq \max_j \|g^{(j)}_t\| < \delta\;.$$ Moreover,
$\psi_1|_{X_j}= \phi|_{X_j}$ for all $j$. The statement of the lemma follows with $U = \sqcup_j X_j$.
\end{proof}

\medskip
\noindent
\begin{lemma}\label{lem-3-mos}
Let $Y_1,...,Y_K \subset M$ be a collection of  balls with pair-wise disjoint closures,
and let $Y:= \sqcup_j Y_j$. Assume that $\mu(Y) < \mu(M)/2$.  Then there exists a volume-preserving
diffeomorhism $S$ of $M$ with $$\text{Closure}(Y) \cap S(\text{Closure}(Y) )= \emptyset\;.$$
\end{lemma}

\begin{proof} The proof is based on a standard application of Moser's argument \cite{Moser}, whose crux
is given by the following statement. Let $V$ be a closed connected manifold equipped with
a pair of volume forms $\mu, \nu$ with $\mu(V)=\nu(V)$. Then there exists a diffeomorphism $T$ of $V$ with $T_*\mu=\nu$. Furthermore, assume that $A \subset M$ is a closed (possibly, disconnected) domain
with smooth boundary such that $M \setminus A$ is connected. If $\mu=\nu$ on $A$,  Moser's method guarantees that one can choose $T$ to be the identity on $A$.

We build the desired diffeomorphism $S$ as follows. Enlarge slightly the balls from $Y$ to get
open subsets $Y \subset Y' \subset Y''$ with
$$\text{Closure}(Y) \subset Y', \;\text{Closure}(Y') \subset Y''\;.$$
We shall assume that $\mu(Y'') < \mu(M)/2$.

Recall that $Y,Y'$ and $Y''$ are collections of $K$ pair-wise disjoint balls
in a connected manifold. It follows that the set $P:= \text{Interior}(M \setminus Y'')$ is connected.
Furthermore, as the diffeomorphism group
of $M$ acts transitively on $K$-tuples of points, there exists a
smooth diffeomorphism $F : M \to M$ with $F(Y') \subset P$. Put $\nu= F_*\mu$.

Choose a volume form $\sigma$ on $M$ with the following properties:
\begin{itemize}
\item $\sigma= \mu$ on $Y''$;
\item $\sigma= \nu$ on $F(Y')$;
\item $\sigma(M)=\mu(M)$.
\end{itemize}
Indeed, the criterion for the existence of such a form is
$\sigma(Y'') + \sigma(F(Y')) < \mu(M)$.  But this holds true since
$$
\sigma(Y'') + \sigma(F(Y'))=
\mu(Y'') + \nu(F(Y')) =
\mu(Y'') +\mu(Y') < 2\mu(Y'') < \mu(M)\;.$$

Apply now Moser's argument to $\sigma$ and $\mu$, and get a diffeomorphism $G: M \to M$
with $G_*\sigma=\mu$. It follows that
\begin{equation}\label{eq-measures}
(GF)_* (\mu|_{Y'}) = G_*(\nu|_{F(Y')}) = G_*(\sigma|_{F(Y')})= \mu|_{G(F(Y'))}\;.
\end{equation}
Furthermore, since $\sigma=\mu$ on $Y''$, the diffeomorphism $G$ produced by Moser's method automatically
equals to the identity on $Y''$. Thus, $G(P) = P$ and hence $G(F(Y')) \subset P$.

It remains to modify $GF$  outside $Y'$ to a $\mu$-preserving diffeomorphism. To this end,
we use again Moser's method and apply it to the volume forms $\theta= (GF)_*\mu$ and $\mu$. We get a diffeomorphism
$H: M \to M$ with $H_*\theta=\mu$. Since  by \eqref{eq-measures} $ \theta= \mu$  on $G(F(Y'))$,
and the latter subset has a connected complement, the diffeomorphism $H$ equals the identity on  $G(F(Y'))$.  Put $S= HGF$. We get that $S(Y') = GF(Y') \subset P$ and $S_*\mu = H_*\theta= \mu$. This completes the proof.
\end{proof}

\medskip
\noindent
{\bf Proof of Lemma \ref{lem-folk}:}
Take the set $W$ and $\epsilon>0,\delta>0$ as in the formulation of the lemma.
Let $Y_1,...,Y_K$ be a collection of  balls with pair-wise disjoint closures  such that $Y_j \subset W$
and for $Y= \sqcup_j Y_j$ holds $\mu(W \setminus Y) < \epsilon/9$ (see Lemma \ref{lem-1-mos}). By Lemma
\ref{lem-3-mos} there exists a volume-preserving diffeomorphism $S$ of $M$ with
$c:= d(Y,S(Y))>0$.

Take $\kappa < \min(\epsilon/9,c)$.
By a result due to Ostrover and Wagner \cite[Lemma 2.4]{OW}, which can be considered as a far reaching generalization of Katok's lemma \cite{Katok},
there exists a Hamiltonian diffeomorphism $\phi$ such that the set
$$\Sigma:= \{x \in M\;:\; d(S(x),\phi(x)) \geq \kappa \}$$
satisfies $\mu(\Sigma) < \kappa$. Put $Z = Y \setminus \Sigma$ and note that
$\mu(W \setminus Z) < 2\epsilon/9$ and $\phi(Z) \cap Z = \emptyset$
since $$d(Z,\phi(Z)) \geq d(Z,S(Z)) -c >0\;.$$

Finally, applying Lemma \ref{lem-2-mos} to $Z$ we get a subset $U \subset Z$ with
$\mu(Z \setminus U) < \epsilon/9$ and a Hamiltonian $g_t$ with $\ell_{cl} < \delta$
such that its time-one-map $\psi$ satisfies $\psi(U)=\phi(U)$.
We end up with $U \subset W$ such that $\psi(U) \cap U =\emptyset$ and $\mu(W \setminus U) < \epsilon/3$.
Therefore,
$$\mu(\psi(W) \cap W) \leq 3\mu(W \setminus U) <  \epsilon  \;,$$
and we are done.
\qed


\medskip

In the next proof we shall use the {\it trace correspondence} which immediately follows
from the property (\ref{eq:estime_diag}): there exists $\delta>0$ such that
\begin{equation}\label{eq-tr-cor}
\Bigl| \text{trace}(T_\hbar(f)) - (2\pi\hbar)^{-n}\int_M f d\mu \Bigr|  \leq \delta \cdot \|f\|_{L_1} \hbar^{-(n-1)}
\end{equation}
for all $f \in \Ci (M)$. Recall that $\mu$ stands for the symplectic volume on $M$.

\medskip
\noindent {\bf Proof of Theorem \ref{thm-main-31}:}
  By outer regularity of $\mu$, there exists an open neighborhood $W$ of $\op{supp} u$ such that $\mu (W) < \mu (M) /2$. Fix $\epsilon' >0$. By Lemma \ref{lem-folk} there exists
a Hamiltonian $f_t$, $t \in [0,1]$ with $\ell_{cl} \leq \delta$ such that its time-one-map $\phi$ satisfies
$\mu(W  \cap\phi(W))\leq \epsilon'$.
Choosing $\epsilon'$ small enough, we can achieve that
\begin{equation}
\label{ex-flex-vsp-1}
\int_M u \cdot u\circ \phi^{-1} d\mu < \bigl( \epsilon/2 \bigr)^2   \;.
\end{equation}
Put $v_\hbar = u/R_\hbar$, where $R_\hbar$ is the Rawnsley function.
By (\ref{eq:estime_diag}), $v_\hbar = (2\pi\hbar)^n u (1+\bigo(\hbar))$ and the derivatives of $v_\hbar$ are $\bigo(\hbar^n)$.
We have $\theta_\hbar= (2\pi\hbar)^n (T_\hbar(u) + \bigo(\hbar))$.

Recall that all the remainders should be understood in the sense of the operator norm. Thus
$\text{trace}\; \bigo(\hbar) = \bigo(\hbar^{-n+1})$
because the dimension is $\bigo ( \hbar^{-n})$.

Denote by $U_{\hbar}$ the time one map of the Schr\"{o}dinger evolution of $T_\hbar(f_t)$. By Egorov theorem \eqref{eq-Egorov},
$$\sigma_\hbar:= U_\hbar\theta_\hbar U_\hbar^{-1} = (2\pi \hbar)^n \left(T_\hbar(u \circ \phi^{-1})+ \bigo(\hbar)\right)\;.$$
Further, by quasi-multiplicativity (\ref{eq:P3})
$$\theta_\hbar\sigma_\hbar= (2\pi\hbar)^{2n} \left(T_\hbar(u)T_\hbar(u \circ \phi^{-1}) + \bigo(\hbar)\right)$$
$$= (2\pi\hbar)^{2n} \left(T_\hbar(u \cdot u \circ \phi^{-1})+ \bigo(\hbar)\right)\;.$$
Next we apply the trace correspondence \eqref{eq-tr-cor} and get
$$I:= \text{trace}(\theta_\hbar \sigma_\hbar)= (2\pi\hbar)^n \text{Vol}(M)^{-1} \int_M u \cdot u \circ \phi^{-1} d\mu + \hbar^n\bigo(\hbar)\;.$$
Now we can estimate the fidelity by formula \eqref{eq-fid-upbd}. We get that for a sufficiently small $\hbar$,
$$\Phi(\theta_\hbar,\sigma_\hbar)\leq (\dim \Hilb_\hbar \cdot I)^{1/2}=
\Biggl( \int_M u \cdot u \circ \phi^{-1} d\mu + \bigo(\hbar) \Biggr)^{1/2} < \epsilon\;,$$
where in the equality we used that $\dim \Hilb_\hbar = \op{vol} (M)/ (2\pi \hbar)^n ( 1 + \bigo ( \hbar))$ and  the second inequality follows from \eqref{ex-flex-vsp-1} above. It remains to note that by the norm correspondence (\ref{eq:P1}),
$\ell_q \leq \ell_{cl} \leq \delta$. This completes the proof.
\qed

\section{Other examples of dislocations} \label{sec:other-exampl-disl}

In this section we will present two examples of $\bigo ( \hbar^{\infty})$-dislocation with a quantum Hamiltonian of the form $s_{\hbar} T_{\hbar} ( f(\cdot, \hbar))$ where $s_{\hbar}$ is a family of positive numbers converging to $0$ in the semiclassical limit and $f( \cdot, \hbar)$ a family of $\Ci ( M)$ admitting an asymptotic expansion of the form $f_0 + \hbar f_1 + \ldots$.  We note that:
\begin{enumerate}
\item the energy of $s_{\hbar} T_{\hbar} ( f(\cdot, \hbar))$ is of order $s_{\hbar}$;
\item  the microsupports of the initial quantum state and its image under the quantum propagator $\exp ( \frac{is_\hbar}{\hbar} T_{\hbar} ( f ( \cdot, \hbar)))$ coincide.
\end{enumerate}
The second fact follows from Remark \ref{rem-hbardepham}.

 The first example  consists on a superposition of coherent states on a $\hbar$-dependent grid with mesh $s_{\hbar}$, and $s_{\hbar} = \hbar^{\delta}$ , $\delta \in (0,1/2)$. In the second example, we will consider any Lagrangian state and $s_{\hbar} = \hbar$, so we will reach the quantum speed limit.

For the Lagrangian states, we will use the definitions and results of \cite{oim2}. The first example is also based on this paper, because we will need some results on the coherent state propagation which follows from the description of the Schwartz kernel of the quantum propagator as a Lagrangian state. Since \cite{oim2} is written in the setting of K\"ahler quantization, we will assume in this section  that  $(M, \om)$ is K\"ahler and the space $\Hilb_{\hbar}$ consists of the holomorphic sections of $L^k \otimes L'$ with $k = 1/\hbar$. Here $L$ and $L'$ are two Hermitian holomorphic line bundles, the curvature of the Chern connection of $L$ being $\frac{1}{i} \om$. Nevertheless, we think that all the constructions we present still hold in the general symplectic case.

\subsection{Superposition of coherent states} \label{subsec-scc}

Choose an open set $U$ of $M$ equipped with Darboux coordinates $(x_i)$, a function $f\in \Ci(M)$ which is equal to $x_1/2$ on $U$ and a function $\varphi \in \Ci (M)$ such that $\int |\varphi|^2 \mu = 1$ and for any $t \in [0,1]$, $\phi_t ( \op{supp} \varphi ) \subset U$ where $\phi_t$ is the Hamiltonian flow of $f$. Of course, for any $x\in  U$, $\phi_t (x) = (x_1, \ldots, x_n, x_{n+1} + t/2, x_{n+2}, \ldots , x_{2n})$.

For any $s \in ]0,1]$, $\hbar \in \Lambda$, define the vector $\Psi_{s,\hbar} \in \Hilb_{\hbar}$ by
\begin{gather} \label{eq:1}
 \Psi_{s,\hbar} = s^{n} \sum_{x \in U \cap ( s \Z^{2n})} \varphi (x)  \xi_{x, \hbar}
\end{gather}
where $\xi_{x,\hbar} = e_{x,\hbar} / \| e_{x,\hbar} \|$ is a normalized coherent state at $x$. Here we have identified $U$ with a subset of $\R^{2n}$ by using the coordinates $(x_i)$.

\begin{thm} \label{theo:disloc_sup_coh_state}
\noindent
\begin{enumerate}
\item
For any $g \in \Ci ( M)$, we have for any $N \geqslant 0$, for any $\hbar \in \Lambda$ and $s \in ]0,1]$
$$ \langle  T_{\hbar} (g) \Psi_{s,\hbar} , \Psi_{s,\hbar} \rangle = \int_{M} g |\varphi |^2 \; \mu + \bigo ( s^N) + \bigo ( \hbar) + \bigo ( ( s^{-2} \hbar)^{n+N} )$$
where the $\bigo$'s are independent of $\hbar$ and $s$. Furthermore if the supports of $g$ and $\varphi$ are disjoint, then $ \langle  T_{\hbar} (g) \Psi_{s,\hbar} , \Psi_{s,\hbar} \rangle$ is $\bigo ( s^{-2n} \hbar ^{\infty})$.
\item
If $U_{\hbar}(t)$ denotes the Schr\"odinger flow of $T_{\hbar} ( f)$, we have for any $N \geqslant 0$, for any $\hbar \in \Lambda$ and $s \in (0,1]$
$$ \langle U_{\hbar}(s) \Psi _{s, \hbar}, \Psi_{s, \hbar} \rangle = \bigo ( (s^{-2}\hbar ) ^{n+N})$$
where the $\bigo$ is independent of $\hbar$ and $s$.
\end{enumerate}
\end{thm}

To explain qualitatively this result, recall that the normalized coherent state $\xi_{x,\hbar}$ is concentrated on a ball centered at $x$ with radius of order $\sqrt \hbar$.
So in the sum (\ref{eq:1}), the various $\xi_{x, \hbar}$ do not interact when $s^{-2} \hbar$ is small. Hence. for the first assertion, we only have to estimate  $s^{2n} \sum |\varphi (x)| ^2 \langle T_{\hbar} (g) \xi_{x,\hbar}, \xi_{x, \hbar}\rangle$, which is done via the Berezin transform and a Riemann sum.
For the second assertion, we use first that the propagated coherent state $U_{\hbar} (s) \xi_{x,\hbar}$ is concentrated on a ball centered at $\phi_s(x)$ with radius of order $\sqrt \hbar$ and second that $\phi_s$ displaces uniformly the set $(\op{supp} \varphi ) \cap s \Z^{2n}$. This explain the half factor in the definition of $f$.

Let us apply Theorem \ref{theo:disloc_sup_coh_state} with a scale factor $s$ depending on $\hbar$:  choose $s_{\hbar } = \hbar ^{\delta}$ with $\delta \in (0,1/2)$.  We deduce from the first part that when $\hbar$ is sufficiently small, $\|\Psi_{s_{\hbar}, \hbar} \| \neq 0$. Furthermore, the semiclassical measure of the normalized state $$\widetilde{\Psi}_{\hbar} :=\Psi_{s_\hbar,\hbar}/ \| \Psi_{s_\hbar, \hbar} \|$$ converges to $|\varphi|^2 d \mu$. If the support of $g$ does not meet the support of $\varphi$, $ \langle  T_{\hbar} (g) \widetilde\Psi_{\hbar} , \widetilde\Psi_{\hbar} \rangle =\bigo (  \hbar ^{\infty})$. Thus
$\op{MS}(\widetilde\Psi_{\hbar}) \subset \text{supp}\; \varphi$. By Corollary \ref{cor}, which provides
the opposite inclusion, the microsupport of  $\widetilde \Psi_{ \hbar}$ coincides with the support of $\varphi$.

By the second part of Theorem \ref{theo:disloc_sup_coh_state}, the quantum Hamiltonian $s_\hbar T_{\hbar} (f)$ $\bigo ( \hbar^{\infty})$-dislocates  $\widetilde{\Psi}_{\hbar}$. The quantum energy of this Hamiltonian is of order $\hbar^{\delta}$. As already noticed, the Schr\"odinger flow at time one of $s_\hbar T_{\hbar} (f)$ preserves the microsupport.

For the proof of Theorem \ref{theo:disloc_sup_coh_state}, we need the following off-diagonal estimates. Introduce a distance $d$ on $M$ which is equivalent in any chart to the Euclidean distance. For any $N \in \N$, we have
\begin{gather} \label{eq:off_diag_Toeplitz}
  \bigl| \langle  T_{\hbar}(g) \xi_{x,\hbar} , \xi_{y,\hbar} \rangle \bigr| = C_N \bigl( e^{-d(x,y)^{2} / (Ch) }+ \hbar^{N} \bigr)
\end{gather}
and denoting by $\phi_t$ the Hamiltonian flow of $f$,
\begin{gather} \label{eq:off_diag_propagator}
 \bigl| \langle  U_{\hbar} (t)  \xi_{x,\hbar} , \xi_{y,\hbar} \rangle  \bigr| = C_N' \bigl( e^{-d(\phi_{t}(x),y)^{2} / (C'h) }+ \hbar^{N} \bigr)
\end{gather}
Here the constants $C_N$ and $C$ only depend on $g$, $C_N'$ and $C'$ on $f$. Estimates (\ref{eq:off_diag_Toeplitz}) and (\ref{eq:off_diag_propagator}) hold for any $x,y \in M$, $\hbar \in \Lambda$ and $t \in [0,1]$.

To get these estimates, observe that for any operator $A$ of $\Hilb_{\hbar}$, the Schwartz kernel of $A$ at $(x,y)$ is $\langle A e_{x,\hbar}, e_{y, \hbar} \rangle$. Then  (\ref{eq:off_diag_Toeplitz}) and (\ref{eq:off_diag_propagator}) follow from the description of the Schwartz kernels of Toeplitz operators and their quantum propagators  in \cite{oim1} and \cite{oim2} respectively.
\begin{lemma} \label{lem:estim_sum_exp}
For any subset $A$ of $\Z^{2n}$, there exists $C>0$ such that for any positive numbers $s,\lambda$, we have
\begin{gather} \label{eq:estim_sum_exp}
 \sum _{x \in s ( \Z^{2n}\setminus A) } e^{-\lambda |x|^2} \leqslant C (s^2 \lambda)^{-n} e^{ - s^2\lambda   t_A}
\end{gather}
where $t_A = \min \{ | x|^2 ; \; x \in \Z^{2n} \setminus A \}$.
\end{lemma}
Consequently, if $0 \in A$, then $t_A \geqslant 1 $ and the left-hand side of (\ref{eq:estim_sum_exp}) is a $\bigo ( ( s^2 \lambda)^{-(n+N)})$ for any $N \geqslant 0$.

\begin{proof}[Proof of Theorem \ref{theo:disloc_sup_coh_state}]
Assume first that the supports of $g$ and $\varphi$ are disjoint. Then by (\ref{eq:estime_outside_diag}), for any $x \in \op{supp} \varphi$, $\| g \xi_{x,\hbar} \| = \bigo ( \hbar^{ \infty})$, uniformly with respect to $x$.  The cardinal of $\op{supp} \varphi \cap s \Z^{2n}$  being a $\bigo ( s^{-2n})$, we get that $\|g \Psi_{s, \hbar}\|$ is a $\bigo ( s^{-n}\hbar^{\infty})$. Since $\|\xi_{x,\hbar} \| =1$, we have $\| \Psi_{s, \hbar} \| =\bigo ( s^{-n})$. By Cauchy-Schwarz inequality, it comes that   $\langle g \Psi_{s,\hbar}, \Psi_{s,\hbar} \rangle =  \bigo ( s^{-2n} \hbar^{\infty})$.

Consider now any function $g$. Denoting as in Section \ref{sec:semiclassical-limit} by $\mathcal{B}_{\hbar} ( g)$ its Berezin transform, we have for any $x \in U \cap s \Z^{2n}$
\begin{xalignat}{2} \label{eq:dec}
\begin{split}
 \langle g \xi_{x,\hbar}, \Psi_{s,\hbar} \rangle & = s^n \biggl( \overline{\varphi (x)} \mathcal{B}_{\hbar}(g)(x) + \sum_{ y \in  ( U \cap s \Z^{2n}) \setminus \{ x \} }  \overline{\varphi (y)} \langle g \xi_{x,\hbar} , \xi_{y,\hbar} \rangle \biggr)  \\
& = s^n \bigl( \overline{\varphi (x)} g(x)   + \bigo ( \hbar) + \bigo ( (s^{-2}\hbar)^{n+N} )
\end{split}
\end{xalignat}
for any $N \geqslant 0$, where we have used the estimate (\ref{eq:Berezin_transform}) of the Berezin transform, Lemma \ref{lem:estim_sum_exp} and (\ref{eq:off_diag_Toeplitz}). We have also used that $\bigo (\hbar^N)$ is contained in $\bigo ( (s^{-2} \hbar)^N)$ because $s \leqslant 1$. Since $\varphi$ is compactly supported in $U$, we deduce from Euler-Maclaurin formula with integral remainder \cite[Theorem 10.2]{HaWa} that
$$ s^{2n} \sum_{x \in s \Z^{2n} } |\varphi ( x) |^2 = \int_U |\varphi (x) |^2 \; d \mu (x) + \bigo ( s^{\infty}) .$$
Multiplying by $s^{-n}$ and summing over $x \in   U \cap s \Z^{2n}$ in (\ref{eq:dec}), we get
$$ \langle g \Psi_{s,\hbar}, \Psi_{s,\hbar} \rangle = \int_M |\varphi|^2 g \mu  + \bigo (s^\infty) + \bigo ( \hbar ) + \bigo (  (s^{-2}\hbar)^{n+N} ) ,$$
which proves the first part of the result.

For the second estimate, observe that when $s$ is sufficiently small, for any $x \in (\op{supp} \varphi) \cap s \Z^{2n}$, the Euclidean distance in the chart $(U, x_i)$ between $\phi_s (x)$ and $U \cap  s \Z^{2n}$ is $s/2$. We deduce  from (\ref{eq:off_diag_propagator}) and Lemma \ref{lem:estim_sum_exp}   that for any $x \in U \cap s \Z^{2n}$,
$$ \langle U_{\hbar}(s) \xi_{x,\hbar}, \Psi_{s,\hbar} \rangle = s^n \bigo ( (s^{-2} \hbar)^{n+N}) $$
where the $\bigo $ is uniform with respect to $x$. We conclude immediately.
\end{proof}

\subsection{Lagrangian states dislocation}\label{subsec-lasd}

Let $N$ be a closed connected Lagrangian submanifold of $M$. Since the curvature of $L$ is proportional to the symplectic form, the bundle $L \rightarrow N$ is automatically equipped with a flat connection. Let us assume that this bundle admits a flat (with respect to this connection) nowhere vanishing section $s$. We normalize $s$ so that its pointwise norm is equal to 1. To the pair $(N,s)$ is associated a space $\mathcal{S} (N,s)$ of {\em Lagrangian states} defined as follows. An element $\Psi$ of $\mathcal{S} ( N , s)$ is a family $( \Psi_{\hbar} \in \Hilb_{\hbar})$ of the form
\begin{gather} \label{eq:5}
 \Psi_{\hbar} = k^{n/4} E^k g(\cdot , k ) + \bigo ( k^{-\infty})
\end{gather}
where
\begin{enumerate}
\item $E$ is any section of $L$ such that its restriction to $N$ is $s$, its pointwise norm is $<1$ outside $N$ and for any $Z \in \Ci ( M , T^{1,0} M)$ the covariant derivative of $s$ with respect to $\con{Z}$ vanishes to infinite order along $N$.
\item $g( \cdot, k)$ is a family of sections of $L' \rightarrow M$ having an asymptotic expansion of the form $g_0 + k^{-1} g_1 + \ldots $ for the $\Ci$ topology.
\end{enumerate}
Since $|E|<1 $ outside $N$, the pointwise norm of $\Psi_{\hbar}$ is a $\bigo ( \hbar^{\infty})$ outside $N$. Furthermore, we have on $N$
$$ \Psi_{\hbar} ( x) = k^{n/4 } s^{k} (x) \Bigl(  g_0 ( x) + \bigo ( k^{-1} ) \Bigr) , \qquad x \in N$$
The restriction of $g_0$ to $N$ is called the {\em principal symbol} of the Lagrangian states $(\Psi_{\hbar})$.

$\mathcal{S} (N,s)$ is not empty. Actually, by Proposition 2.4 of \cite{oim2}, for any sequence $(g_{\ell})$ of $\Ci ( M, L')$, there exists a Lagrangian state satisfying the conditions given above. Of course, the assumption that $L \rightarrow N$ admits a flat non trivial section is rather restrictive. We can easily relax this condition by asking that for some $k_0$, $L^{k_0} \rightarrow N$ admits a flat non trivial section, the corresponding Lagrangian state would be defined only for $k \in k_0 \N^{*}$.

For instance, consider the projective space $M=\C P ^n$ with $L$ being the dual of the tautological bundle. For any positive $k$, $H^0 ( M , L^k)$ identifies with the space $\C_k [z_0, \ldots, z_n]$ of homogeneous polynomials of degree $k$. Introduce the polytope $P = \{  t \in \R_{\geqslant 0 } ^{n+1} /   t_0 + \ldots + t_n = 1 \}$ and denote by $P^0$ its interior. Then to any $t \in P^0$ is associated a Lagrangian torus $$N_t = \{ [z_0: \ldots : z_n ] / |z_i|^2 = t_i, \forall i \; \} \subset \C P^n.$$
Choose $t \in P^0 \cap \Q^{n+1}$. Let $k_0$ be such that $\ell = k_0 t \in \N^{ n+1}$ and for any $m \in \N^*$, set $P_m ( z ) = z_0 ^{m\ell_0} \ldots z_n^{m \ell_n} \in H^0 ( M, L^{m k_0})$. We claim that the sequence $( P_m/ \| P_m\|)$ is a Lagrangian state associated to $N_t$. The proof is an exercise: first, one checks that the pointwise norm of the section $P_0$ attains its maximum on $N_t$ and that the restriction of $P_0$ to $N_t$ is flat. Second, one has to estimate the norm $\| P_m \|$.

More generally, natural examples of Lagrangians section are given by joint eigenvectors of quantum integrable systems, cf. \cite{oim2}. The associated Lagrangian manifolds are leaves of a Lagrangian foliation, each leaf corresponding to a different eigenvalue. In these cases, instead of working with a fixed Lagrangian submanifold, it is convenient to describe uniformly the Lagrangian states corresponding to the various leaves. Nevertheless this description is more delicate because we can not use a globally defined section $E$ as above. Indeed, the leaves admitting a global flat section form a discrete set. Actually, we have a local representation completely similar to (\ref{eq:5})
and these local sections can be glued together only when a Bohr-Sommerfeld condition is satisfied.
To avoid these complications,  we will restrict ourselves to the spaces $\mathcal{S} ( N,s)$ introduced above. But all the results we will prove could be extended to this generalized setting.

In the sequel, we will need the two following basic properties. Consider two Lagrangian states $(\Psi_\hbar)$ and $(\Psi_{\hbar}')$ in $\mathcal{S} (N,s)$ with principal symbol $\sigma, \sigma' \in \Ci ( N, L')$. Then by Propositions  2.6 and 2.7 of \cite{oim2}
\begin{enumerate}
\item $ \langle \Psi_{\hbar} , \Psi'_{\hbar} \rangle $ has an asymptotic expansion of the form $\sum_{\ell=0}^{\infty} \hbar ^{\ell} a_\ell$ where the $a_{\ell}$'s are complex number, the leading order term being
$$ a_0 =  \int_N ( \sigma , \sigma') \nu   $$
where $(\sigma , \sigma')$ is the pointwise scalar product and $\nu$ is a nowhere vanishing density of $N$ independent of $(\Psi_{\hbar})$, $(\Psi'_{\hbar})$.
\item for any $f \in \Ci (M)$, $(T_{\hbar} (f) \Psi_{\hbar} )$ belongs to $\mathcal{S} (N,s)$, its principal symbol is $f|_N \si$.
\end{enumerate}
The density $\nu$ can be computed explicitly, but we don't need it for what we will do.

\begin{thm} \label{theo:disl-lagr}
Assume that $N$ is diffeomorphic to $S^1 \times N'$. Let  $(\Psi_{\hbar})$ be a Lagrangian state in $\mathcal{S} ( N, L)$ such that $\| \Psi_{\hbar} \| = 1$ and whose principal symbol does not vanish anywhere. Then there exists a family $(f_{\hbar})$ of $\Ci (M)$ having an asymptotic expansion
$$ f_{\hbar} = \hbar ( f_0 + \hbar f_1 + \hbar ^2 f_2 + \ldots )$$
for the $\Ci$ topology and such that the Schr\"odinger flow of the quantum Hamiltonian $T_{\hbar} ( f_\hbar)$ $\bigo ( \hbar^{\infty})$-dislocates $\Psi_{\hbar}$.
\end{thm}

For the proof, we need the following lemma.
\begin{lemma} \label{lem:coef}
For any nowhere vanishing density $\delta$ of $S^1$, there exists a function $ f \in \Ci (S^1, \R)$ such that
\begin{itemize}
\item [{(i)}] we have  $ \int_{S^1} e^{if }  \delta = 0$
\item [{(ii)}] for any $ z \in \C$, there exists $ g \in \Ci ( S^1, \R)$ such that $ \int_{S^1} e^{if} g \delta = z$.
\end{itemize}
\end{lemma}

\begin{proof}
Renormalizing $\delta$, we can assume that $\int_{S^1} \delta = 2 \pi$. Choose a diffeomorphism $S^1 \simeq \R / (2 \pi \Z)$ such that $\delta = |d \theta|$ where $\theta $ denotes the coordinate of $\R$. Choose a smooth function $h : (0,\pi) \rightarrow \R$ compactly supported, taking its values in $[0,\pi]$ and such $ h( t) = \pi$ if $t \in [\pi/5, 4\pi/5]$. Let $$I_s  = \int_0 ^{\pi} \cos ( s h (\theta) ) d \theta.$$ Then $I_{1/2} >0$, $I_1 <0$. So by continuity there exists $s \in (1/2,1)$ such that $I_s = 0$. Now define the function $f$ so that it is odd and coincides with $sh$ in $[0,\pi]$. The function $\theta \rightarrow \sin f( \theta) $ is odd, so its integral on $S^1$ vanishes. The function $\theta \rightarrow \cos  f ( \theta)$ is even so that $ \int_{S^1} \cos  f ( \theta) \; d \theta = 2 \int _0 ^{\pi} \cos  f ( \theta) \; d \theta = 0 $. Consequently the integral of $\exp ( i f)$ vanishes, which proves (i).

Observe that $f ( \theta ) = s \pi $ on $[\pi/5, 4\pi/5]$. Since $s \in ( 1/2,1)$, $\cos (s\pi)$ and $\sin (s\pi)$ don't vanish. Choose a function $g_1\in \Ci ((0,\pi), \R)$ supported in $( \pi/5, 4\pi/5)$  and such that $ \int_0^\pi g_1\;d\theta = 1/(2 \cos ( s\pi))$. Extend $g_1$ to an even function of $S^1$. By adapting the above argument, we get
$$ \int_{S^1}  g_1( \theta) \cos (f( \theta ) )\;d\theta =  1, \qquad \int_{S^1} g_1 (\theta) \sin ( f( \theta) ) d\theta = 0.$$
Similarly choose $g_2 \in \Ci ((0,\pi), \R)$ supported in $( \pi/5, 4\pi/5)$  and such that $ \int_0^\pi g_2\; d\theta = 1/(2 \sin ( s\pi))$. Extend $g_2$ to an odd function of $S^1$. Then
$$ \int_{S^1}  g_2 ( \theta) \cos (f( \theta ) )\;d\theta = 0, \qquad \int_{S^1} g_2 (\theta) \sin ( f( \theta) ) d\theta = 1\;.$$ The function $g= \text{Re}(z) g_1 + \text{Im}(z)g_2$
satisfies (ii), which completes the proof.
\end{proof}

Consider a family $F_{\hbar}$ of $\Ci ( M)$ having an asymptotic expansion of the form $f_0 + \hbar f_1 + \ldots$ for some functions $f_\ell \in \Ci ( M)$, $\ell \in \N$. The following result  is completely standard and can be deduced  from the functional calculus for Toeplitz operators or from the property (E).

\begin{lemma} $\bigl ( \exp ( i T_{\hbar} ( F_{\hbar}) \bigr)$ is a Toeplitz operator with principal symbol $e^{if_0}$ meaning that for any  $p \in \N$, we have
$$ \exp \bigl( i T_{\hbar} ( F_{\hbar}) \bigr) = T_{\hbar} ( e^{if_0} ( 1 + \hbar g_1+ \ldots +  \hbar^{p} g_{p} ))  + \bigo ( \hbar^{p+1})$$
for some functions $g_\ell$, $\ell \in \N^*$. Furthermore,
\begin{gather} \label{eq:approx_exp_op}
 \exp \bigl( i T_{\hbar} ( F_{\hbar}) \bigr) = ( 1 + i \hbar^{p+1} T_{\hbar} (f_{p+1}) ) \exp \bigl( i T_{\hbar} ( F^p_{\hbar}) \bigr) + \bigo ( \hbar^{p+2})
\end{gather}
where $F^p_{\hbar} = f_0 + \hbar f_1 + \ldots + \hbar^p f_p$.
\end{lemma}

\begin{proof}[Proof of Theorem \ref{theo:disl-lagr}]
Assume first that $M$ is two dimensional and $N$ is a circle. We will prove by induction on $p$ that if the functions $f_0$, \ldots, $f_p$ are conveniently chosen, then
\begin{gather} \label{eq:hyp_rec}
  \bigl\langle \exp( i T_{\hbar} ( F_{\hbar})) \Psi_{\hbar} , \Psi_{\hbar} \bigr\rangle = \bigo ( \hbar ^{p+1})
\end{gather}
For $p=0$, we use that $\exp( i T_{\hbar} ( F_{\hbar}))= T_{\hbar} ( e^{if_0}) + \bigo ( \hbar)$ so that
\begin{gather} \label{eq:lot}
 \bigl\langle \exp( i T_{\hbar} ( F_{\hbar})) \Psi_{\hbar} , \Psi_{\hbar} \bigr\rangle = \int_N e^{if_0} ( \si , \si ) \nu  + \bigo ( \hbar)
\end{gather}
where $\si$ is the principal symbol of $\Psi_{\hbar}$. By Lemma \ref{lem:coef} applied to $\delta = ( \si , \si ) \nu$, there exists $f_0$ such that the integral in (\ref{eq:lot}) vanishes.

Assume now that (\ref{eq:hyp_rec}) holds  for some $p \in \N$. By (\ref{eq:approx_exp_op}), $\exp( i T_{\hbar} ( F_{\hbar})) = \exp( i T_{\hbar} ( F^p_{\hbar})) + \bigo ( \hbar^{p+1})$ so
$$ \bigl\langle \exp( i T_{\hbar} ( F^p_{\hbar})) \Psi_{\hbar} , \Psi_{\hbar} \bigr\rangle = \bigo ( \hbar ^{p+1}) = \hbar^{p+1} z + \bigo ( \hbar^{p+2}) $$
Since $T_{\hbar} (f_{p+1}) \exp( i T_{\hbar} ( F^p_{\hbar})) = T_{\hbar} ( f_{p+1} e^{if_0}) + \bigo ( \hbar)$ we deduce from (\ref{eq:approx_exp_op}) that
\begin{xalignat*}{2}
 &  \bigl\langle \exp( i T_{\hbar} ( F_{\hbar})) \Psi_{\hbar} , \Psi_{\hbar} \bigr\rangle \\  & = \bigl\langle \exp( i T_{\hbar} ( F^p_{\hbar})) \Psi_{\hbar} , \Psi_{\hbar} \bigr\rangle + i \hbar^{p+1} \bigl\langle T_{\hbar} ( f_{p+1} e^{if_0})   \Psi_{\hbar} , \Psi_{\hbar} \rangle   + \bigo ( \hbar^{p+2})   \\
& = \hbar^{p+1} \Bigl( z+ i \int_{S^1} f_{p+1} e^{if_0} ( \si, \si ) \nu   \Bigr) + \bigo ( \hbar^{p+2})
\end{xalignat*}
By Lemma \ref{lem:coef}, we can choose $f_{p+1}$ so that the coefficient of $\hbar^{p+1}$ vanishes, which completes the proof for $N$ being a circle.

The proof for $N = S^1 \times N'$ is completely similar. Let $\pi$ denote the projection $ N \rightarrow S^1$. We choose the coefficients $f_{\ell}$ such that they are constant on the fibers of $\pi$. We apply Lemma \ref{lem:coef} to $\delta = \pi_{*}  \bigl( (\si, \si) \nu \bigr)$. So for any $h \in \Ci(S^1)$, we have $\int_N (\pi^{*} h) (\si, \si) \nu  = \int_{S^1} h \delta$.
\end{proof}

\section{Dislocation vs. displacement on small scales}\label{sec-ddsc}

In the present section we shall fix a {\it test ball} $B \subset M$, that is an open ball whose closure lies in a Darboux chart equipped with coordinates $(x_1,...,x_{2n})$.
The ball $B$ is given by $\{\sum x_i^2 < 1\}$.
In the chart the symplectic form $\omega$ is given by $dx_1 \wedge dx_2 +...+dx_{2n-1} \wedge dx_{2n}$.

Fix a time-dependent Hamiltonian $\con{f}_t$ compactly supported in $B$ and define the rescaled Hamiltonian
$$ f_{t,s} (x) =  s^2\overline{f}_t(x/s), \qquad s\in [0,1] .$$
Observe that the corresponding quantum Hamiltonian $T_{\hbar} ( f_{t,s})$ satisfies
$$\ell_q (T_{\hbar} ( f_{t,s})) \leqslant s^2 \| \overline{f} \| .$$
Similarly, fix a classical state $\overline{\tau}$ supported in $B$ and let $\tau_{s}$ be the classical state obtained by contracting $\tau$ by a factor of $s$.
So $$Q_\hbar ( \tau_s) = \int_B P_{sx,\hbar } \; d\con{\tau} (x).$$
In what follows, we denote by $\bigo( (s^{-2} \hbar)^{\infty})$ any quantity depending on $s$ and $\hbar$ which is a $\bigo( (s^{-2} \hbar )^{N})$ for any $N\geqslant 0$.

\begin{thm}\label{thm-disloc-21}
Assume that the time-one-map of the Hamiltonian flow of $\overline{f}_t$ displaces the support of $\overline{\tau}$. Then  the Schr\"{o}dinger flow generated by $ T_\hbar(f_{t,s})$ $a_{\hbar,s}$-dislocates the state $Q_\hbar(\tau _{s})$ with $a_{\hbar,s}=  \bigo( (s^{-2} \hbar) ^\infty)$.
\end{thm}

It is interesting to apply this result with $s$ depending on $\hbar$. Assume that $(\overline{f}_t)$ displaces the support of $\overline{\tau}$.
\begin{enumerate}
\item Choose $s_{\hbar} = \hbar^{\epsilon}$ with $\epsilon \in [0,1/2)$. Then $ \bigo( (s_{\hbar} ^{-2} \hbar) ^\infty) = \bigo ( \hbar^{\infty})$ and we conclude that $ T_\hbar(f_{t,s_{\hbar}})$ $\bigo(\hbar^{\infty})$-dislocates $Q_\hbar(\tau _{s_{\hbar}})$ .
\item Choose $s_{\hbar} = r \hbar^{1/2} $. For any $a \in(0,1)$, if $r$ is sufficiently large and $\hbar$ sufficiently small,   $ T_\hbar(f_{t,s_{\hbar}})$ $a$-dislocates $Q_\hbar(\tau _{s_{\hbar}})$.  Observe that  $\ell_q ( T_{\hbar} ( f_{t,s}) ) \leqslant r^2 \| \overline{f}\| \hbar$. So the energy is of the same order as the quantum speed limit.
\end{enumerate}

When $\overline{\tau}$ is of class $\Cl^3$, we have the following converse. Write $\overline{\tau} = \overline{u}d\mu$, where $\overline{u} \in \Cl^3_c (B)$
and $\int_B \overline{u}d\mu=1$.
Choose a family $(s_\hbar)$ in $(0,1]$ and a positive number $\lambda$.
\begin{thm}\label{thm-disloc-1}
Assume that the quantum Hamiltonian $F_{t,\hbar}= T_\hbar(f_{t,s_\hbar})$ $a_\hbar$-dislocates the quantum state $Q_\hbar(\tau_{s_{\hbar}} )$. If $s_\hbar^{-2}\hbar$, $a_\hbar/(s_\hbar^{-2}\hbar)^n$ and $\hbar$ are small enough,
the time-one-map of $f_{t,s_{\hbar}}$ displaces the set $s_\hbar \cdot \{\overline{u} > \lambda\}$.
In particular,
\begin{equation}
\label{eq-disloc-1}
\begin{split}
\ell_q(F_\hbar) \geq e_B(s_\hbar \cdot \{\overline{u} > \lambda\})+ \bigo(\hbar) \\ \geq
s_\hbar^2 \cdot e_{\R^{2n}}(\{\overline{u} > \lambda\})(1+ \bigo(s_\hbar^{-2}\hbar))\;.
\end{split}
\end{equation}
\end{thm}

Here we use the notation $s \cdot X$ for the dilation of the subset $X \subset \R^{2n}$, see \eqref{eq-dil}. The second inequality follows from the scaling law \eqref{eq-des} for the displacement
energy and the obvious inequality $e_B(X) \geq e_{\R^{2n}}(X)$.

Let us illustrate Theorem \ref{thm-disloc-1} by setting
$s_\hbar \geq C\hbar^{\varepsilon}$ with $\epsilon \in [0,1/2)$. Then
we get that $\ell_q \gtrsim \hbar^{2\epsilon}$.
One can phrase this as follows:  on the scales exceeding the wave length scale,
the semiclassical speed limit is ``significantly
more restrictive" (in terms of the asymptotics in $\hbar$) than the universal quantum speed limit.

For the proof of Theorem \ref{thm-disloc-21} we will need the following improvement of (\ref{eq:estime_outside_diag}) and (\ref{eq:expansion_product}). Choose a distance $d$ on $M$ equivalent in any chart to the Euclidean distance. Then we have for any $x,y \in M$, $\hbar \in \Lambda$ and $N \in \N$
\begin{gather} \label{eq:estime_outside_diag_exponential} \tag{CS3}
 | e_{x, \hbar} (y) | \leqslant C_{N} \bigl( \hbar^{-n} e^{- d(x,y)^2/ ( C \hbar)} +  \hbar^{N} \bigr)
\end{gather}
where $C$ and $C_N$ do not depend on $x,y,\hbar$.
Note that this bound is obtained from \eqref{eq:off_diag_Toeplitz} by substituting $g=1$.

Furthermore the bidifferential operators $B_{\ell}$ and remainders $r_{N, \hbar}$ appearing in (\ref{eq:expansion_product}) satisfy
\begin{xalignat}{2} \label{eq:Expansion_complement} \tag{E2}
\begin{split}
& \text{- }  B_{\ell} \text{ is a bidifferential operator of order } 2 \ell \\
& \text{- }   r_{N, \hbar} (f,g) = \bigo ( \hb^{ N +1 }) |f,g|_{2N+2}
\end{split}
\end{xalignat}
These two properties hold for any K\"ahler quantization or its generalization to symplectic manifolds. 
(\ref{eq:Expansion_complement}) is proved in \cite[Theorem 1.6 of the arXiv version]{oim_symp}. The similar result for the Bargmann space was proved in \cite{our_paper}.

The proof of Theorem \ref{thm-disloc-1} is based on (\ref{eq:P1}), (\ref{eq:P2}) and (\ref{eq:P3}).

\subsection{Proof of Theorem \ref{thm-disloc-1}}\label{subsec-proof-disloc}

The proof is similar to the one of Theorem \ref{thm-intro-main-1}, cf. Section \ref{sec:proof-theorem-1.2ii}.
In what follows we write $c_1,c_2,...$ for positive constants independent on $\hbar$.

One has $d \tau_{s_\hbar} = u_{\hbar} d \mu$ where $u_{\hbar} ( x) = s_{\hbar}^{-2n} \overline{u}(x/ s_{\hbar})$. So
$$Q_{\hbar} ( \tau_{s_{\hbar}}) = T_\hbar ( v_{\hbar}) \qquad \text{ with }   v_{\hbar} = u_{\hbar} / R_{\hbar} .$$
Set $g_{\hbar} = v_{\hbar} / \| v_{\hbar}\|$ and recall the constant $b(g,f)$ and $c(f)$ defined in (\ref{eq-b-defin}) and (\ref{eq-c-defin}).

\begin{lemma} \label{lem:3}
$b( g_\hbar, ( f_{t, s_\hbar})) \leqslant c_1 s_{\hbar}^{-2} $ and $c( f_{t,s_{\hbar}} ) \leq c_2$.
\end{lemma}

\begin{proof}
We say that a family $(j_{\hbar}, \hbar \in \Lambda)$ is a symbol of order $N$ if there exists $C>0$ such that for $p =0,1,2,3$ and $\hbar \in \Lambda$, one has
$$ |j_{\hbar} |_p \leqslant C s_{\hbar}^{-N-p} $$
If the family depends also on $t \in [0,1]$, we require that $C$ does not depend on $t$. Observe that $(u_{\hbar})$ is a symbol of order $2n$. By property (\ref{eq:estime_diag}),
$$|\hbar^{-n} R_{\hbar}^{-1}|_p \leqslant C_p.$$
Consequently, $( \hbar^{-n} R_{\hbar}^{-1} u_{\hbar}  )$ is also a symbol of order $2n$. Since
\begin{gather} \label{eq:equiv_norm_vhbar}
\| v_{\hbar} \| \sim ( 2 \pi s_{\hbar}^{-2} \hbar)^n \| \overline{u} \|,
\end{gather}
it comes that $g_{\hbar}$ is a symbol of order $0$.

The Hamiltonian flow $\phi_{t,s}$ of $(f_{t, s})$ is given by $\phi_{t,s} (x) = s \overline{\phi}_t (x/s)$, with $\overline{\phi}_t$ the Hamiltonian flow of $\overline{f}_t$. So we have for any $p \in \N$
\begin{gather} \label{eq:estim_flow}
|\phi_{t,s}^{-1}|_p \leqslant C_ps^{1-p} \qquad \forall t \in [0,1]\;
\end{gather}
From this, one deduces that $(g_{\hbar} \circ \varphi_{t,s_{\hbar}}^{-1})$ is also a symbol of order $0$. Using now that $(f_{t, s_{\hbar}})$ is a symbol of order $-2$, we conclude easily.
\end{proof}

Let $U_{\hbar}$ be the time-one-map of the Schr\"odinger flow of $T_{\hbar} ( f_{t, s_\hbar})$. Set $\Gamma_{q,\hbar}:= \Gamma_q ( \theta_\hbar, U_\hbar\theta_\hbar U_\hbar^*)$, where
$\theta_\hbar:= Q(\tau_{s_\hbar})$.

\begin{lemma} \label{lem:4}
$\Gamma_{q,\hbar} \leqslant c_3 a_{\hbar} ( s_{\hbar}^{-2} \hbar ) ^{-n} $ if $ s_{\hbar}^{-2} \hbar$ is sufficiently small.
\end{lemma}
\begin{proof}
Using (\ref{eq:P1}), the equivalence (\ref{eq:equiv_norm_vhbar}) and the fact that $(\hbar^{-n}v_\hbar) $ is a symbol of order $0$, we obtain that
$$ \| \theta_{\hbar} \|_{op} = ( 2 \pi s_\hbar^{-2} \hbar)^{n} \| \overline{u} \| \bigl( 1 + \bigo ( s_{\hbar}^{-2} \hbar )  \bigr).
$$
So if $ s_{\hbar}^{-2} \hbar$ is sufficiently small, $\|\theta_\hbar\|_{op} \geq c_4 (s_\hbar^{-2}\hbar)^n$. Proposition \ref{prop-fidel} yields
$ \Ga_{q,\hbar} \leqslant a_{\hbar}/ ( c_4 (s_{\hbar}^{-2} \hbar)^n.$
\end{proof}

 By Corollary \ref{cor-main-disloc},
the Hamiltonian $f_{t,s_\hbar}$ displaces the set
$X:=\{g_\hbar > A\}$ and
$$
\ell_{q,\hbar} \geq e_B (g_\hbar,A) -c_2 \hbar
$$
with
$$A = \sqrt{\Gamma_{q,\hbar}+ 3b_\hbar\hbar} \leq \sqrt{c_4  a_\hbar ( s_\hbar^{-2} \hbar)^{-n}  + c_5 (s_\hbar^{-2} \hbar) } \;.$$
Recalling the definition of $g_\hbar$, we readily get that
$X$ contains the set $$s_\hbar \cdot \{\overline{u} > A'\}, \;\; \text{where}\;\;A' = A \cdot \|\overline{u}\|(1+\bigo(\hbar))\;.$$
If now $\hbar$, $ a_\hbar ( s_\hbar^{-2} \hbar)^{-n} $  and $s_\hbar^{-2} \hbar$  are small enough, we achieve that $A' < \lambda$.
It follows that the time-one-map of $f_{t,s_\hbar}$ displaces the set
$s_\hbar \cdot \{\overline{u} > \lambda \}$, and furthermore by \eqref{eq-main-disloc} and \eqref{eq-des}
\begin{equation}
\label{eq-disloc-11}
\begin{split}
\ell_q \geq e_B(s_\hbar \cdot \{\overline{u} > \lambda\}) - c_2 \hbar \geq e_{\R^{2n}}(s_\hbar \cdot \{\overline{u} > \lambda\})- c_2\hbar =
\\
s_\hbar^2 \cdot e_{\R^{2n}}(\{\overline{u} > \lambda\})(1-c_6(s_\hbar^{-2}\hbar))\;.
\end{split}
\end{equation}
This completes the proof.
\qed

\medskip
\noindent
\begin{rem}\label{rem-assumptions-pf}{\rm The above proof holds for any quantization
satisfying properties (\ref{eq:estime_diag}), (\ref{eq:P1}), (\ref{eq:P2}), and (\ref{eq:P3}). Observe also Theorem \ref{thm-disloc-1} still holds if we consider instead of $(f_{t,s_{\hbar}})$ any Hamiltonian $(\tilde{f}_{t,\hbar})$ satisfying for $p=0,1,2,3$, $t \in [0,1]$,
$$ | \tilde{f}_{t,\hbar} |_p \leqslant C s_{\hbar}^{2 -p} , \qquad |\tilde{\phi}_{t,\hbar}^{-1}|_p \leqslant C s_{\hbar}^{1-p} $$
where $\tilde{\phi}_{t,\hbar}$ the Hamiltonian flow of $(\tilde{f}_{t,\hbar})$.
}
\end{rem}

\subsection{Proof of Theorem \ref{thm-disloc-21}} \label{sec:proof-theorem-disloc_21}

The proof is similar to the proof of Theorem \ref{thm-former-main-i}.  However, we won't formalize the notion of microsupport at small scale even if the proof is based on this idea. In the proof, we will use an auxiliary function $\con{g} \in \Ci (M)$ supported in $B$ with value in $[0,1]$ and such that $\con{g}=1$ on a neighborhood of $\op{supp} \overline{\tau}$ and the Hamiltonian flow of $\con{f}_t$ displaces the support of $\con{g}$. Let $g_{s} ( x) = \con{g}( x / s)$.

As in the formulation of Theorem \ref{thm-disloc-21}, denote by $\tau_{s}$ the rescaled measure.


\begin{lemma} \label{lem:small-scale-cutoff}
We have
$$  \bigl\| T_\hbar (g_{s } ) Q_\hbar ( \tau_{s}) - Q_{\hbar} ( \tau _{s}) \bigr\|_{tr} =   \bigo ( ( s^{-2} \hbar)^{\infty} )   .$$
\end{lemma}

The proof is based on (\ref{eq:estime_outside_diag_exponential}).

\begin{proof}
Let $j_s = 1 - g_s$. Let $r = s \hbar^{-1/2} $. Let $x\in \op{supp} \con{\tau}$. We will first prove that there exists $C'>0$ such that
\begin{gather} \label{eq:9}
 \| j_s e_{sx,\hbar} \| ^2 =  \bigo ( \hbar^{\infty}) + \bigo ( \hbar^{-n} e^{- r^2 / C' } )
\end{gather}
where the $\bigo$'s are uniform with respect to $x$.

Let $K$  be a compact set of $B$ containing $s (\op{supp}\con{\tau}) $ for any $ s\in [0,1]$. Applying (\ref{eq:estime_outside_diag}) to the disjoint closed sets $K$ and $M \setminus B$, we get:
\begin{gather} \label{eq:8}
 \int_{M \setminus B} |j_s(y)|^2 | e_{sx, \hbar}(y) |^2 \; d\mu(y) =   \int_{M \setminus B}  | e_{sx, \hbar}(y) |^2 \;d \mu (y)=  \bigo ( \hbar^{\infty})
\end{gather}

Let $j \in \Ci ( \R^{2n})$ given by $j(y) = 1 - \con{g}(y) $ if $y \in B$, and $j( y) =1 $ otherwise. For any $y \in B$, we have $ j_s ( y) = j ( y/s)$. Denoting by $\mu_L$  the Lebesgue measure of $\R^{2n}$, we have
\begin{gather*}
 \hbar^{-n} \int_B |j_s (y) |^2 e^{-2   |s x - y |^2 / ( C \hbar  ) } \; d \mu_L (y)   = r^{2n} \int_{s^{-1}B} |j(y)|^2 e^{-2  r^2 |x - y |^2   / C  } \; d\mu_L ( y) \\
 \leqslant  r^{2n} \int_{\R^{2n} } |j(y)|^2 e^{-2  r^2 |x - y |^2  /C  } \; d \mu_L ( y)
\end{gather*}
Since  $\con{g} = 1 $ on a neighborhood of $\op{supp} \con{\tau}$, there exists $c>0$ not depending on $x$ such that  $j(y) \neq 0$ implies $|x - y | \geqslant c$. So $j(y) e^{ -  r^2 |x-y|^2/C} \leqslant e^{-r^2/C' }$ with $C'= C/c^2$. Thus the previous integral is bounded above by
\begin{gather*}
 r^{2n} e^{-r^2/C'} \int_{\R^{2n}}  e^{- r^2 |x - y |^2/C } \; d\mu_L ( y) = e^{-r^2 / C'}  \int_{\R^{2n}}  e^{- |y |^2/C } \; d \mu_L ( y)
\end{gather*}
Using the estimate (\ref{eq:estime_outside_diag_exponential}), we obtain
\begin{gather} \label{eq:7}
  \int_{ B} |j_s (y)|^2 | e_{sx, \hbar} (y)|^2 \; d \mu(y) = \bigo ( \hbar^{\infty}) + \bigo ( \hbar^{-n} e^{-r^2 / C' }) \end{gather}
Finally Equations (\ref{eq:8}) and (\ref{eq:7}) imply (\ref{eq:9}).

Recall that $P_{z,\hbar}$ is the  orthogonal projector onto the line spanned by $e_{z, \hbar}$. For any endomorphism $T$ of $\Hilb_{\hbar}$, $\| T P_{z,\hbar} \|_{tr} = \| T  e_{z,\hbar} \| /   \| e_{z,\hbar} \|$.
 Since by (\ref{eq:estime_diag}),  $\| e_{z,\hbar} \| \sim (2\pi \hbar )^{-n/2}$, we deduce from (\ref{eq:9}) that for any $x \in \op{supp} ( \con{\tau})$,
 $$  \bigl\| T_{\hbar}( j_s) P_{sx,\hbar } \bigr\|_{tr}  = \bigo ( \hbar^{\infty}) + \bigo ( e^{- r^2 / (2 C')  }) = \bigo ( r^{-\infty}) $$
uniformly with respect to $x$. Here we have used that $\bigo ( \hbar^{\infty})$ is contained in $\bigo ( r^{-\infty})$ because $s \leqslant 1$ so that $\hbar \leqslant r^{-2}$.

Since $Q_\hbar ( \tau_s) = \int_M P_{sx,\hbar } \; d\con{\tau} (x)$, we have
$$ T_\hbar ( j_s) Q_\hbar ( \tau_s) = \int_M T_\hbar ( j_s) P_{sx,\hbar } \; d\con{\tau} (x), $$
and we get the conclusion by integrating the last estimate.
\end{proof}

 We say that a family $j = (j_s,
s \in (0,1])$ of $\Ci (M )$ is a symbol of order $m$ if
\begin{gather} \label{eq:symbol_def}
 | j_s |_p \leqslant C_p s ^{-m -p}, \qquad \forall p \in \N .
\end{gather}
We will also consider symbols depending on time $t \in [0,1]$. In that case, we require that (\ref{eq:symbol_def}) holds uniformly with respect to $t$. These symbols are similar, but different from the ones used in the proof of Lemmas \ref{lem:3} and \ref{lem:4}.

Recall that $f_{t,s}(x)= s^2 \overline{f}_t(x/s)$. Let $\phi_{t, s }$ be the corresponding Hamiltonian flow and $U_{s,\hbar} (t)$ be the Schr\"odinger flow generated by $T_\hbar ( f_{t, s})$.

\begin{lemma} \label{lem:small_scale_egorov}
There exists a family of symbols $(g^\ell_s)$, $\ell \in \N^*$, of order $2 \ell$, such that for any $N$, $s\in (0,1]$ and $\hbar \in \La$, we have
$$ U_{s,\hbar} T_\hbar(g_{s}) U_{s,\hbar}^* = T_\hbar(g_{s} \circ \phi_{s}^{-1} ) + \sum_{\ell=1 }^{N} \hbar^{\ell} T_\hbar( g^{\ell}_s  \circ \phi_{s} ^{-1} ) + \bigo ( ( s^{-2} \hbar)^{N+1})  $$
where $\phi_{s} = \phi_{1,s}$ and $U_{s,\hbar} = U_{s,\hbar}(1)$.
\end{lemma}

The proof is based on the following consequence of (\ref{eq:expansion_product}), (\ref{eq:Expansion_complement}):
\begin{gather} \label{eq:2}
\begin{split}
 (i\hbar)^{-1}  [ T_\hbar(f), T_\hbar(g) ]   = &  T_\hbar(\{ f, g\})  + \sum_{\ell = 1} ^{ N} \hbar^{\ell} T_\hbar ( P_\ell ( f,g))   \\ & +  \bigo ( \hbar^{N +1} ) |f,g|_{2N + 4}
\end{split}
\end{gather}
where the $P_{\ell}$ are bidifferential of order $2\ell +2$.

\begin{proof}
We define inductively a family $(g^\ell_{t , s } , t\in [0,1], s \in (0,1], \ell \in \N)$ of $\Ci (M)$ by $g^0_{t,s} = g_s$ and for any $\ell \in \N^*$,
\begin{gather} \label{eq:10}
 g^\ell_{ \tau ,s} = \sum_{m = 1 }^{\ell}  \int_0^\tau P_{m} ( f_{t,s} , g^{\ell-m}_{t, s } \circ \phi_{t,s}^{-1} )  \circ \phi_{t,s} \;  dt
\end{gather}
We check that $(g^\ell_{t,s})$ is a symbol of order $2 \ell $. To do this, we use the following facts:
\begin{enumerate}
\item $(f_{t,s})$ is a symbol of order $ -2$, $(g_s)$ is a symbol of order $0$.
\item  for any $p \in \N$,
$ | \phi_{t,s} |_p \leqslant C_p s^{ 1-p}$ and $| \phi_{t,s}^{-1} |_p \leqslant C_p s^{ 1-p}$.
\item if $j_{t,s}$ is a symbol, then $j_{t,s} \circ \phi_{t,s} $ and $j_{t,s} \circ \phi_{t,s} ^{-1}$ are symbols of the same order.
\item if $j$, $j'$ are symbol of order $m$, $m'$ and $Q$ is a bidifferential operator of order $q$, then $Q( j, j')$ is a symbol of order $m+m'+ q$.
\end{enumerate}
Then we deduce that $| f_{t,s} , g^\ell_{ t,s} \circ \phi_{t,s}^{-1}|_{2 (N-\ell)+ 4} $ is a $\bigo( s^{ -2 (N+1)})$. By (\ref{eq:2}) we obtain
\begin{gather} \label{eq:3}
\begin{split}    & ( i\hbar)^{-1} [ T_{\hbar} ( f_{t,s}), T_{\hbar} ( g^\ell_{t,s} \circ \phi_{t,s}^{-1} ) ]  - T_{\hbar} ( \{ f_{t,s} ,g^\ell_{ t,s} \circ \phi_{t,s}^{-1} \} ) \\ & = \sum_{m=1}^{N-\ell} \hbar^{m} T_{\hbar} (P_m ( f_{t,s} , g^\ell_{t,s} \circ \phi_{t,s}^{-1}  ))  + \hbar^{-\ell} \bigo ( (  s^{-2} \hbar)^{N+1})
\end{split}
\end{gather}
Fix $N \in \N$ and set $j_{t,s, \hbar} = \sum_{\ell = 0}^{N} \hbar^{\ell} g^\ell_{t,s}$. We have by (\ref{eq:10}) that
\begin{gather} \label{eq:4}
 \Bigl(\frac{d}{dt} j_{t,s,\hbar} \Bigr) \circ \phi_{t,s}^{-1} = \sum_{\ell = 1}^{N} \hbar^{\ell} \sum_{m=1}^{\ell} P_m( f_{t,s} , g^{\ell - m}_{t,s}\circ \phi_{t,s}^{-1})
\end{gather}
We can now compute:
\begin{gather*}
\begin{split}
   \frac{d}{dt} \bigl( U_{s,\hbar}(t)^* T_{\hbar}   ( j_{t,s,\hbar} \circ \phi_{t,s}^{-1} ) U_{s,\hbar}(t) \bigr)  = U_{s,\hbar}(t)^* \Bigl( \frac{i}{\hbar} [ T_{\hbar} (f_{t,s}) , T_{\hbar} ( j_{t,s,\hbar} \circ \phi_{t,s}^{-1} ) ] \\
+ T_{\hbar} \Bigl( \{ f_{t,s} ,  j_{t,s,\hbar} \circ \phi_{t,s}^{-1} \} +  \Bigl(\frac{d}{dt} j_{t,s,\hbar} \Bigr) \circ \phi_{t,s}^{-1} \Bigr)  \Bigr) U_{s,\hbar}(t)
 = \bigo ( ( s^{-2} \hbar)^{N+1})
\end{split}
\end{gather*}
where we have used (\ref{eq:3}) and (\ref{eq:4}).
Integrating this inequality, we get the result with $g^\ell_s = g^\ell_{1,s}$.
\end{proof}

\begin{lemma} \label{lem:product}
$ T_\hbar (g_{s}) U_{s,\hbar} T_\hbar (g_{s}) U^*_{s,\hbar} = \bigo ( ( s^{-2} \hbar )^{\infty})$.
\end{lemma}
This lemma is a consequence of Lemma \ref{lem:small_scale_egorov} and estimate (\ref{eq:Expansion_complement}).

\begin{proof} Since the flow of $\overline{f}_t$ displaces the support of $\overline{g}$, $g_{s}  $ and $g_{s} \circ \phi_{s}^{-1}$ have disjoint supports. By the definition (\ref{eq:10}) of $g^{\ell}_s = g^\ell_{1,s}$,                  the support of $g^\ell_s $ is contained  in the support of $g_{s}$. So $(g_s)$ and $(g^\ell_s \circ \phi_s^{-1})$ have disjoint supports. The result follows now from (\ref{eq:Expansion_complement}), by using that $(g_s)$ and $(g^\ell_s \circ \phi_s^{-1})$ are symbols of order $0$ and $2\ell$ respectively.
\end{proof}

We can now complete the proof of Theorem \ref{thm-disloc-21}.
Write $\theta_{s,\hbar} = Q_{\hbar} ( \tau _{s})$ and $\si_{s,\hbar} =  U_{s,\hbar} \theta_{s,\hbar} U^*_{s,\hbar}$.
\begin{lemma}\label{lem:1}
$\bigl\| ( T_\hbar ( g_s ) - \op{id} ) \theta_{s,\hbar}^{1/2} \sigma_{s,\hbar}^{1/2} \bigr\|_{tr}  = \bigo((s^{-2} \hbar  )^{\infty})$
\end{lemma}
\begin{proof}
Denoting by $\| \cdot \|_{HS}$ the Hilbert-Schmidt norm and using H\"older inequality for Schatten norms, one has
\begin{xalignat*}{2}
 \bigl\| ( T_\hbar ( g_s ) - \op{id} ) \theta_{s,\hbar}^{1/2} \sigma_{s,\hbar}^{1/2} \bigr\|_{tr} & \leqslant  \bigl\| ( T_\hbar ( g_s ) - \op{id} ) \theta_{s,\hbar}^{1/2} \bigr\|_{HS} \| \sigma_{s,\hbar}^{1/2} \|_{HS} \\
 & =  \Bigl( \op{tr} \bigl(  ( T_\hbar ( g_s ) - \op{id} ) \theta_{s,\hbar} ( T_\hbar ( g_s ) - \op{id}) \bigr) \Bigr)^{1/2}\\
 & \leqslant \Bigl(  \|  ( T_\hbar ( g_s ) - \op{id} ) \theta_{s,\hbar} \|_{tr} \| T_\hbar ( g_s ) - \op{id} \|_{op} \Bigr)^{1/2}
\end{xalignat*}
To conclude, we use that $ \| T_\hbar ( g_s ) \|_{op} \leqslant \| g_{s} \| \leqslant 1$ and Lemma \ref{lem:small-scale-cutoff}.
\end{proof}

\begin{lemma} \label{lem:2}
$\bigl\| T_\hbar ( g_s ) \theta_{s,\hbar}^{1/2} \sigma_{s,\hbar}^{1/2} \bigr\|_{tr}  = \bigo(( s^{-2} \hbar  )^{\infty})$
\end{lemma}
\begin{proof}
First, one has by H\"older inequality
$$\| T_\hbar ( g_s ) \theta_{s,\hbar}^{1/2} \sigma_{s,\hbar}^{1/2} \|_{tr} \leqslant \| T_\hbar ( g_s ) \|_{HS}  \| \theta_{s,\hbar}^{1/2} \sigma_{s,\hbar}^{1/2} \|_{HS} =  \| T_\hbar ( g_s ) \|_{HS} \bigl( \op{tr} (\theta_{s,\hbar} \sigma_{s,\hbar} ) \bigr)^{1/2}$$
Since $\int_M |g_s| d \mu = s^{2n} \int_M |g| \mu$, one deduces from (\ref{eq:tracenorm}) that the trace norm of $T_\hbar ( g_s )$ is a $\bigo (( s^{-2} \hbar  )^{-n})$. The Hilbert-Schmidt norm being dominated by the trace norm, we get
$$  \| T_\hbar ( g_s ) \|_{HS} = \bigo ((  s^{-2} \hbar  )^{-n}).$$
To complete the proof, we show  that $ \op{tr}(\theta_{s,\hbar} \sigma_{s,\hbar} )  =    \bigo (( s^{-2} \hbar  )^{\infty}  )$.
Since $\| \sigma_{s,\hbar} \| _{tr} =  1$, one has by submultiplicativity of the trace norm
\begin{gather*} \bigl\|   \theta_{s,\hbar}  ( T_\hbar ( g_s ) - \op{id} )  \sigma_{s,\hbar} \bigr\|_{tr} \leqslant \bigl\| \theta_{s,\hbar}   ( T_\hbar ( g_s ) - \op{id} )  \|_{tr} =  \bigo ((  s^{-2} \hbar )^{\infty}  )
\end{gather*}
by Lemma \ref{lem:small-scale-cutoff}. Similarly, one shows that
$$  \bigl\|   \theta_{s,\hbar}  T_\hbar ( g_s )  U_{s,\hbar}  ( T_\hbar ( g_s ) - \op{id} ) \theta_{s,\hbar} U_{s,\hbar}^* \bigr\|_{tr} =  \bigo ((  s^{-2} \hbar  )^{\infty}  )  .$$
These two estimates imply that
\begin{xalignat*}{2}
 \op{tr}  ( \theta_{s,\hbar} \sigma_{s,\hbar}  ) &  = \op{tr} \bigl(  \bigl(  \theta_{s,\hbar}  T_{\hbar} ( g _s)  \bigr)  U_{s,\hbar} \bigl( T_{s,\hbar} (g_{s} ) \theta_{s,\hbar}  \bigr) U_{s,\hbar}^* \bigr) + \bigo ((s^{-2}\hbar  )^{\infty}  )\\
& = \op{tr} \bigl( \theta_{s,\hbar} \bigl(  T_{\hbar} ( g _{s})   U_{s,\hbar}  T_{\hbar} (g_s ) U_{s,\hbar}^*  \bigr)  \sigma_{s,\hbar}   \bigr)  + \bigo ((  s^{-2}\hbar  )^{\infty}  ) \\
& = \bigo (( s^{-2} \hbar  )^{\infty}
\end{xalignat*}
by lemma \ref{lem:product} and the fact that $\| \theta_{s,\hbar} \sigma_{s,\hbar} \|_{tr} \leqslant 1$.
\end{proof}
Theorem \ref{thm-disloc-21} follows from Lemmas \ref{lem:1} and \ref{lem:2}.

\section{Miscellaneous proofs}
Here we collect proofs of some results used in the paper.

\subsection{Quantum speed limit}\label{subsec-mis-qsl}

First, let us derive inequality \eqref{eq-qsl-univ} from the work of Uhlmann \cite{Uh}. Given a quantum
Hamiltonian $F_t \in \cL(\Hilb)$, $t \in [0,1]$, and a state $\theta \in \cS(\Hilb)$, write $\theta_t$ for the Schr\"{o}dinger evolution of $\theta$. It has been proved in \cite{Uh} that
\begin{equation}\label{eq-uhlmann}
I:= \int_0^1 \sqrt{\op{trace}(F_t^2\theta_t) - (\op{trace}(F_t\theta_t))^2} dt \geq \arccos(a)\hbar \;,
\end{equation}
where $a= \Phi(\theta,\theta_1)$.
By H\"{o}lder inequality,
$$\op{trace}(F_t^2\theta_t) \leq ||F_t^2||_{op} \cdot ||\theta_t||_{tr} = ||F_t||^2_{op}\;.$$
Thus $\ell_q(F) = \int_0^1 ||F_t||dt \geq I$, and \eqref{eq-qsl-univ} follows.

\medskip

Next, let us mention that while the above argument is based on Uhlmann's inequality \eqref{eq-uhlmann}
whose proof is quite tricky, the quantum speed limit for pure states admits the following elementary
proof (cf. \cite{Uh,AH}). Assume  that $\theta=\xi$ is a pure state. The trajectory $\xi(t)$ under the flow $U(t)$ generated by the Hamiltonian $F_t$ satisfies
$$\dot{\xi}(t)=  -\frac{i}{\hbar}F_t \xi(t)\;.$$
We estimate the length $b$ of this trajectory in the projective space $\mathbb{P}(\Hilb)$
with respect to the Fubini-Study metric
by
$$b = \int_0^1 \sqrt{|\dot{\xi}|^2 - |\langle \xi,\dot{\xi}\rangle|^2} dt \leq \int_0^1 |\dot{\xi}|dt \leq \hbar^{-1}\ell_q(F)\;.$$
On the other hand, the Fubini-Study distance between $\xi$ and $U\xi$ equals
$$c= \arccos (|\langle \xi, U \xi \rangle |) \geq \arccos (a)\;.$$
Since $b \geq c$, we get that $$\ell_q(F) \geq  \arccos (a)\hbar$$ which yields \eqref{eq-qsl-univ} for a pure state.

\medskip
For the sake of completeness, we present an elementary proof of a weaker version of \eqref{eq-qsl-univ} for a mixed state,
\begin{equation}\label{eq-qsl-univ-1}
\ell_q(F) \geq  \frac{1}{2}\arccos(a)\hbar\;.
\end{equation}
Note that for the purposes of the present paper the factor $1/2$ is not essential.
Assume that $\theta \in \cS(\Hilb)$ is a general (mixed) quantum state. Consider
the space $\Hilb'$ of all linear operators $\Hilb \to \Hilb$ equipped with the scalar
product $\langle A,B \rangle' := \text{trace}(AB^*)$. We consider $\sqrt{\theta}$
as a pure quantum state in $\cS(\Hilb')$ (called {\it a purification} of $\theta$).
Observe that for every pair of quantum states $\theta,\sigma \in \cS(\Hilb)$
\begin{equation}\label{eq-fidel-vs-scalar}
|\langle \theta, \sigma \rangle'| = \text{trace}(\sqrt{\theta}\sqrt{\sigma}) \leq \Phi(\theta,\sigma)\;. \end{equation}
Furthermore, since the evolution $\sqrt{\theta_t}$ of the purification of $\theta$ is given by
$$\sqrt{\theta_t}= U(t) \sqrt{\theta} U(t)^{-1}\;,$$ the Hamiltonian $F'_t \in \cL(\Hilb')$ of this evolution is given
by $$F'_t(A) = [F_t,A]\;\; \forall A \in \Hilb'\;.$$
Looking at the action of $F'_t$ in the eigenbasis of $F_t$, one readily checks that
\begin{equation}\label{eq-norm-purif}
\|F'_t\|_{op} \leq 2\|F_t\|_{op}\;.
\end{equation}
It remains to notice that if $F_t$ $a$-dislocates a mixed state $\theta \in \cS(\Hilb)$, then
by \eqref{eq-fidel-vs-scalar} $F'_t$ $a$-dislocates its purification in $\cS(\Hilb')$. Thus inequality
\eqref{eq-qsl-univ-1} follows from \eqref{eq-norm-purif} and the quantum speed limit for pure states.

\medskip

We conclude this section with a couple of remarks illustrating the notion of dislocation.

\medskip
\noindent
\begin{rem} \label{rem-disloc-1}{\rm $ $
\begin{itemize}
\item[{(i)}] If $d \geq 2$, every pure state $\xi$ can be $0$-dislocated with
$\ell_q = \frac{\pi}{2}\hbar$, i.e., the universal speed limit with $a=0$ is sharp for pure states. Indeed, take any unit vector $\eta$ orthogonal to $\xi$, and dislocate $\xi$ by a Hamiltonian which equals
\[\frac{\pi }{2}\hbar \cdot   \left( \begin{array}{cc}
0 & -i  \\
i & 0  \\
\end{array} \right)\]
in the plane generated by $\xi,\eta$ and vanishes in its orthogonal complement. The same argument shows
that every  every mixed state can be $0$-dislocated with  $\ell_q = \frac{\pi}{2}\hbar$ provided $\dim \text{Image}(\theta) \leq d/2$.
\item[{(ii)}] Not every state can be $a$-dislocated with $a <1$: take $\theta=d^{-1}\id$,
where $d=\dim \Hilb$.
\item[{(iii)}]Assume that $\theta$ can be $a$-dislocated, i.e. $\Phi(\theta, U\theta U^{-1}) \leq a$ for some unitary transformation $U$. We claim that then it can be $a$-dislocated with $\ell_q \leq \pi \cdot \hbar$. Indeed, fix an orthonormal
basis in $\Hilb$ and write $U$ as the diagonal matrix with the entries $\exp(i c_j)$, where
$c_j \in [-\pi,\pi)$. The Hamiltonian $F$ given by the diagonal matrix with the entries
$-\hbar c_j$ generates $U$, and the claim follows.
\end{itemize}
}
\end{rem}

\subsection{Estimate of trace norm} \label{sec:estimate-trace-norm}

Here we prove Equation (\ref{eq:tracenorm}).
Without loss of generality, we may assume that $f$ is real valued so that $T_\hbar (f)$ is self-adjoint. Since the map sending $f$ to $T_{\hbar} (f)$ is positive, we have that $ | T_{\hbar} (f) | \leqslant T_{\hbar} ( |f|)$, so that
\begin{gather} \label{eq:est_tr_1}
\| T_{\hbar} (f) \| _{\op{tr}} \leqslant \op{trace} (T_{\hbar} (|f|)).
\end{gather}
Applying (\ref{eq:defToep_eq}) to $|f|$ and taking the trace, we get
\begin{xalignat}{2} \label{eq:est_tr_2}
\begin{split}
 \op{trace}( T_{\hbar}( |f| )) & = \int_M |f|(x) \op{trace} (S_{\hbar,x}) d \mu(x)
\\  &  \leqslant \| f \|_{L_1} \sup_{x \in M}  \| e_x \| ^2 \\
&  \leqslant (2 \pi \hbar)^{-n}   \| f \|_{L_1}( 1  + \bigo ( \hbar) )
\end{split}
\end{xalignat}
by using that $\op{trace} (S_x) = \| e_x \|^2$ and estimate (\ref{eq:estime_diag}). Collecting estimates (\ref{eq:est_tr_1}) and (\ref{eq:est_tr_2}), we obtain the second inequality of (\ref{eq:tracenorm}).

We have for any self-adjoint operator $S$ of $\Hilb_{\hbar}$, for any $x$, $|\langle S e_x, e_x\rangle | \leqslant \langle |S| e_x , e_x \rangle$. By (\ref{eq:defToep_eq}), $\op{id} = \int_M S_{\hbar, x} d\mu (x)$. Composing by $|S|$ and taking the trace, we obtain
$$ \| S \|_{\op{tr}} = \int_M \langle |S| e_x , e_x \rangle d \mu (x)$$
This proves that
$$  \int_M |\langle S e_x , e_x \rangle|  d \mu (x) \leqslant \| S \|_{\op{tr}}. $$
 Applying this to $S  = T_{\hbar} (f)$, we get
$$ \| \mathcal{B}_\hbar(f) R_\hbar \|_{L_1} \leqslant  \| T_{\hbar} (f) \| _{\op{tr}}$$ where $\mathcal{B}_\hbar$ is the Berezin transform and $R_\hbar$ the Rawnsley function. We deduce the first inequality of (\ref{eq:tracenorm}) by using (\ref{eq:estime_diag}) and (\ref{eq:Berezin_transform}).

\medskip
\noindent
{\bf Acknowledgments.} We thank Yael Karshon, Yohann Le Floch, St\'ephane Nonnenmacher, Dominique Spehner and Alejandro Uribe for
useful discussions. We are grateful to the anonymous referee for numerous helpful comments and critical remarks.

\bigskip

\noindent
\begin{tabular}{ll}
Laurent Charles & Leonid Polterovich \\
 UMR 7586, Institut de Math\'{e}matiques  & Faculty of Exact Sciences \\
de Jussieu-Paris Rive Gauche &  School of Mathematical Sciences\\
Sorbonne Universit\'{e}s, UPMC Univ Paris 06 & Tel Aviv University \\
F-75005, Paris, France & 69978 Tel Aviv, Israel \\
laurent.charles@imj-prg.fr & polterov@post.tau.ac.il\\
\end{tabular}

\end{document}